\documentclass[12pt,draft]{IEEEtran}
\usepackage{color}     
\usepackage{epsf,psfrag,amssymb,amsfonts,latexsym,cite,verbatim,enumerate,subfigure}
\usepackage{srcltx,amsmath,cases, graphicx}
\usepackage{times, multirow, array}
\usepackage[mathscr]{eucal}

\def\psfancypar#1#2{\begingroup\def\par{\endgraf\endgroup\lineskiplimit=0pt}
               \setbox2=\hbox{\large\sc #2}
               \newdimen\tmpht \tmpht \ht2 \advance\tmpht by \baselineskip
               \font\hhuge=Times-Bold at \tmpht
               \setbox1=\hbox{{\hhuge #1}}
               \count7=\tmpht \count8=\ht1
               \divide\count8 by 1000 \divide\count7 by \count8 
               \tmpht=.001\tmpht\multiply\tmpht by \count7 
               \font\hhuge=Times-Bold at \tmpht
               \setbox1=\hbox{{\hhuge #1}}
               \noindent
                \hangindent1.05\wd1
               \hangafter=-2 {\hskip-\hangindent
               \lower1\ht1\hbox{\raise1.0\ht2\copy1}%
                \kern-0\wd1}\copy2\lineskiplimit=-1000pt}

\newcommand{\E}{\mbox{{\rm E}}}

\newcommand{\abf}{\mbox{${\bf a}$}}

\def\boxit#1{\vbox{\hrule\hbox{\vrule\kern3pt
        \vbox{\kern3pt#1\kern3pt}\kern3pt\vrule}\hrule}}

\def\reals{ { {\rm  I \kern-0.15em R }  } }
\def\complex{ {\,{{\rm C} \kern-0.50em \raise0.20ex {  |}}\, }}

\def\mubf{\hbox{\boldmath$\mu$\unboldmath}}

\def\pibf{\hbox{\boldmath$\pi$\unboldmath}}

\def\omegabf{\hbox{\boldmath$\omega$\unboldmath}}
\def\Sigmabf{\hbox{$\bf \Sigma$}}

\def\abf{{\bf a}}

\def\nbf{{\bf n}}
\def\obf{{\bf o}}

\def\sbf{{\bf s}}

\def\xbf{{\bf x}}
\def\ybf{{\bf y}}

\def\xbf{{\bf x}}
\def\ybf{{\bf y}}
\def\Abf{{\bf A}}
\def\Bbf{{\bf B}}

\def\Hbf{{\bf H}}
\def\Ibf{{\bf I}}

\def\Pbf{{\bf P}}

\def\Rbf{{\bf R}}

\def\Cc{{\cal C}}

\def\Hc{{\cal H}}
\def\Ic{{\cal I}}

\def\Nc{{\cal N}}

\def\Pc{{\cal P}}

\def\Rc{{\cal R}}
\def\Sc{{\cal S}}

\def\be{\vskip .3cm \begin{equation}}
\def\ee{\end{equation} \vskip .4cm \noindent}
\newcommand{\R}{\mbox{$\hat {\bf R}_{N}$}}

\def\Rxx{\Rbf_{\ssstyle X\kern-.1em X}}

\let\ssstyle=\scriptscriptstyle

\def\Kout{\setbox1=\hbox{\Huge\bf K}\hbox to
1.05\wd1{\hspace{.05\wd1}
}}
\def\Sout{\setbox1=\hbox{\Huge\bf S}\hbox to 1.05\wd1{\hspace{.05\wd1}
}}

  \ifx\LabelFigloaded\MYundefined\relax
  \else
   \fi

  \def\LabelFigloaded{\relax}

  \chardef\LabelFigCatAt\the\catcode`\@
  \catcode`\@=11

 \let\LabelFigwlog@ld\wlog
 \def\wlog#1{\relax}

 \ifx\\\MYundefined@
    \let\\\relax
 \fi

  \def\ms@g{\immediate\write16}

 \def\N@wif{\csname newif\endcsname }
 \def\Temp@ {\N@wif\ifIN@}
 \ifx\INN@\MYundefined@
    \else \let\Temp@\relax
 \fi
 \Temp@

  \def\IN@{\expandafter\INN@\expandafter}
  \long\def\INN@0#1@#2@{\long\def\NI@##1#1##2##3\ENDNI@
    {\ifx\m@rker##2\IN@false\else\IN@true\fi}%
     \expandafter\NI@#2@@#1\m@rker\ENDNI@}
  \def\m@rker{\m@@rker}
 
  \newtoks\Initialtoks@  \newtoks\Terminaltoks@
  \def\SPLIT@{\expandafter\SPLITT@\expandafter}
  \def\SPLITT@0#1@#2@{\def\TTILPS@##1#1##2@{%
     \Initialtoks@{##1}\Terminaltoks@{##2}}\expandafter\TTILPS@#2@}

 \def\Shifted@@#1#2#3{\setbox0=\hbox{#3}%
   \raise -\dp0\vbox {\kern-#2%
       \hbox {\kern#1\unhbox0\kern-#1}%
           \kern#2}}

 \newcount\gridcount
 \newbox\auxGridbox@ \newbox\hGridbox@ \newbox\vGridbox@
 \newbox\Labelbox@ \newbox\auxLabelbox@
 \newbox\Coordinatebox@
 \newtoks\Labeltoks@
 \newdimen\Wdd@ \newdimen\Htt@
 \newdimen\Wddd@ \newdimen\Httt@
 
 \def\Wr@{\immediate\write16}

 \newdimen\GL@wd
 \def\GridLineWidth#1{\GL@wd=#1}

 \def\gobble#1{}
 \def\EdgeErr@{\Wr@{}%
      \Wr@{\string\Edges\space argument
      1, 10, 100 or 1000 please\string!}%
      }

 \newcount\Edgect@

 \def\Sweepup#1\endSweepup{}

 \def\SetEdges@{%
    \edef\Zr@@s{\expandafter\gobble\number\Edgect@\empty}%
        \count255=0\Zr@@s\relax
        \ifnum\count255=\z@\else\EdgeErr@\show\tailtest\fi
        \count255=1\Zr@@s\relax
        \ifnum\count255=\Edgect@\relax\else\EdgeErr@\show\leadtest\fi
    \EdgGl@b\edef\Zr@s{\expandafter\gobble\Zr@@s\empty}
    \ifnum\Edgect@>\@ne\relax\EdgGl@b\let\L@Dc\empty
        \else\EdgGl@b\edef\L@Dc{\string.}\fi
    \ifnum\Edgect@>\@ne\relax
        \EdgGl@b\edef\Edgescale@##1{\divide##1 by \Edgect@}%
        \else\EdgGl@b\edef\Edgescale@##1{}\fi
    }

 \def\Edges#1{\Edgect@=#1\relax
     \let\EdgGl@b\global \SetEdges@}

 \Edges{1}

 \def\hhrule{\hrule height \GL@wd\vskip-.\GL@wd}

 \def\hRule@{%
   \advance\gridcount -2%
   \vfil\hhrule\vfil
   \llap{\smash{\raise -2.5pt
     \hbox{\L@Dc\number\gridcount\Zr@s\kern2pt}}}%
   \hhrule
   }

\def\vvrule{\vrule width \GL@wd \kern-\GL@wd}

 \def\vRule@{\advance\gridcount 2%
   \hfil\vvrule\hfil
   \setbox\auxGridbox@=\vbox to 0pt
      {\vskip \Htt@\vskip 2pt
        \hbox to 0pt{\hss\L@Dc\number\gridcount\Zr@s\hss}\vss}%
      \wd\auxGridbox@=0pt \box\auxGridbox@
   \vvrule
   }

 \def\PlaceGrid@@{\gridcount=10 
  \setbox\hGridbox@=\hbox{%
        \hbox{%
             \hskip-.4pt\vrule
             \vbox to \Htt@{%
               \offinterlineskip\parindent=\z@\relax
               \hbox to \Wdd@{\hfil}
               \hRule@\hRule@\hRule@\hRule@
               \vfil\hhrule\vfil}%
             \vrule\hskip-.4pt}
    }%
  \gridcount=0%
  \setbox\vGridbox@=\hbox{%
      \vbox{\offinterlineskip\parindent=0pt\hsize=0pt
         \vskip-.4pt\hrule%
         \hbox to \Wdd@{%
                 \vtop to \Htt@{\vfil}%
                 \vRule@\vRule@\vRule@\vRule@
                 \hfil\vvrule\hfil}%
         \hrule\vskip-.4pt}}%
  \wd\hGridbox@=0pt\ht\hGridbox@=0pt
  \wd\vGridbox@=0pt\ht\vGridbox@=0pt
  \hbox{\box\hGridbox@\box\vGridbox@}%
  }

 \def\LabelsGlobal{\def\LabGl@b{\global}}
 \def\LabelsLocal{\def\LabGl@b{}}
 \LabelsGlobal 

 \def\SetLabels#1\endSetLabels{%
   \LabGl@b\Labeltoks@={#1()\\}%
   }

 \LabGl@b\Labeltoks@={()\\}

 \def\ShowGrid{\LabGl@b\let\PlaceGrid@\PlaceGrid@@}
 \def\HideGrid{\LabGl@b\let\PlaceGrid@\relax}
 \def\Grids{\ShowGrid\LabGl@b\let\GridSwitch@\ShowGrid}
 \def\noGrids{\HideGrid\LabGl@b\let\GridSwitch@\HideGrid}

 \noGrids

 \def\bAdjust@@{%
     \setbox\auxLabelbox@=\hbox{\raise \dp\auxLabelbox@
            \box\auxLabelbox@}}
 \def\bAdjust@{\let\vAdjust@\bAdjust@@}

 \def\eAdjust@@{\dimen0=-.5\ht\auxLabelbox@
     \advance\dimen0 by .5\dp\auxLabelbox@
     \setbox\auxLabelbox@=
            \hbox{\raise\dimen0\box\auxLabelbox@}}
 \def\eAdjust@{\let\vAdjust@\eAdjust@@}

 \def\tAdjust@@{%
     \setbox\auxLabelbox@=\hbox{\raise-\ht\auxLabelbox@
            \box\auxLabelbox@}}
 \def\tAdjust@{\let\vAdjust@\tAdjust@@}

 \let\vAdjust@\relax

 \def\lAdjust@{\let\hAdjust@\rlap}
 \def\rAdjust@{\let\hAdjust@\llap}

 \let\hAdjust@\relax\let\vAdjust@\relax

 \def\FetchLabel@#1(#2)#3\\{%
     \IN@0#2@@\ifIN@
        \setbox0=\hbox{\ignorespaces#1#3\unskip}%
        \ifdim\wd0>0pt
           \ms@g{}%
           \ms@g{ !!! Bad label(s)? !!!}%
           \message{ #1(#2)#3}%
        \fi
        \def\LabelMole@##1\endFetchLabel@{%
            \IN@0()\\@##1@%
            \ifIN@\def\Temp@{\FetchLabel@##1\endFetchLabel@}%
            \else\def\Temp@{}%
            \fi
            \Temp@
           }%
     \else
       \ignorespaces#1\unskip
       \setbox\auxLabelbox@=%
         \hbox to 0pt{\hss\ignorespaces\hAdjust@
          {\ignorespaces#3\unskip}\hss}%
       \vAdjust@
       \let\hAdjust@\relax\let\vAdjust@\relax
       \AugmentLabelBox@@{#2}%
       \ht\Labelbox@=0pt\dp\Labelbox@=0pt
       \let\LabelMole@\FetchLabel@%
     \fi\LabelMole@}

 \newtoks\XYSep@ 
 \def\SetXYSeparator#1{%
     \IN@0#1@@\ifIN@\XYSep@{*}%
     \else
     \XYSep@{#1}%
     \fi
     }

 \SetXYSeparator*

 \def\AugmentLabelBox@@#1{%
     \IN@0\the\XYSep@ @#1@\ifIN@
       \SPLIT@0\the\XYSep@ @#1@%
       \setbox\Labelbox@=\hbox to 0pt{%
         \unhbox\Labelbox@
         \Shifted@@{\the\Initialtoks@\Wddd@}%
         {\the\Terminaltoks@\Httt@}%
         {\box\auxLabelbox@}}%
     \else
         \ms@g{}%
         \ms@g{ !!! Bad insertion point. !!!}%
         \message{ (#1\ this point was rejected.)}%
     \fi
    }

 \def\FetchOption@#1[#2]#3\endFetchOption@{%
    \def\temp{#1}
    \ifx\temp\empty
       \Edgect@=#2\relax
       \let\EdgGl@b\relax
       \SetEdges@
       \Cleaner@#3%
    \fi}

 \def\Cleaner@#1[@]{\Labeltoks@{#1}}
     
 \def\PlaceLabels@@{\mathsurround=0pt
     \def\Cr@{\\}%
     \let\L\lAdjust@\let\R\rAdjust@
     \let\B\bAdjust@\let\E\eAdjust@\let\T\tAdjust@
     \expandafter\FetchOption@\the\Labeltoks@[@]\endFetchOption@
     \Wddd@=\Wdd@ \Edgescale@\Wddd@ 
     \Httt@=\Htt@ \Edgescale@\Httt@
     \expandafter\FetchLabel@\the\Labeltoks@\endFetchLabel@
     \box\Labelbox@
     }%

 \let \PlaceLabels@\PlaceLabels@@

 \def\AffixLabels#1{\setbox\Coordinatebox@=\hbox{#1}%
      \Wdd@=\wd\Coordinatebox@ \Htt@=\ht\Coordinatebox@
      \advance\Htt@ \dp\Coordinatebox@
      \hbox{\copy\Coordinatebox@\kern-\Wdd@ 
           \Shifted@@{0pt}{-\dp\Coordinatebox@}%
           {\PlaceLabels@\PlaceGrid@}%
           \kern\Wdd@}%
      \GridSwitch@ 
      \LabGl@b\Labeltoks@{()\\}%
      }
 
   \let\wlog\LabelFigwlog@ld   
   \catcode`\@=\LabelFigCatAt  

                                By

              Raymond S\'eroul <A18645@FRCCSC21.BITNET>
                                and 
              Laurent Siebenmann <lcs@topo.math.u-psud.fr>
    
              VERSIONS: July 1991, Oct 1991, Jan 1992, July 1992

INTRODUCTION

      This labelling package is intended for TeX users who
rely on non-TeX sources for for their graphics inserts.  It
provides means for adding TeX labels to such inserts with a
minimum of fuss. 

       For most labels, TeX users have in the past found it
reasonably convenient to rely on non-TeX sources. Typical
occasions when an inescapable need for TeX labels seemed to
arise are

 (a) when the graphics program lacks certain exotic or complex
mathematical symbols

 (b) when the very highest typographical quality is wanted for the
labels

 (c) when labels included with the graphics fail to print, 
 and you cannot figure out why (cf. boxedeps.doc).  The labels
 provided by labelfig.tex are 100

       Since this package first appeared, many users, who in the
past scarcely dreamed of using TeX labels, have come to use
nothing but.  So it is now appropriate to add

Intoxication Warning:  TeX labels may be addictive and expensive. 

     If you have a fast preview you may disagree, and even find
that this package provides an agreeable paste-up environment; see
extra applications at end.

     Note to publishers: It is possible and convenient to ultimately
export the TeX labels produced by labelfig.tex to become an integral
part of the EPS file. This is often desired by a publisher who typically
uses an "upmarket" graphics or page layout program, with which the
staff is skilled in perfecting figures.  See Appendix I for
a recipe.

     The authors are grateful to Patrick Ion of Math Reviews for
helpful comments and encouragement.

BASIC INSTRUCTIONS

    After reading in the macro file using

preview or proof your figure with a coordinate grid printed on
top, by typing the following:

    \ShowGrid  
    \AffixLabels{<the graphics insertion>}

Here <the graphics insertion> is what you would type to insert
the graphics object alone without the grid.  This must provide
for the space around it. For example <the graphics insertion>
might well be \BoxedEPSF{MyFigure scaled 700} using the
boxedeps.tex macro package (from same source); this provides a
TeX box containing the encapsulated PostScript insert specified by
the file MyFigure. \AffixLabels{...} provides the grid (supposing
\ShowGrid is present) and later, once you have specified labels
using the grid, it will "tack on" the labels.

     The grid is a sort of (usually elongated) checkerboard of
ten rows and ten columns and its (internal) partitions are by
default numbered  .1, ... ,.9  both horizontally (X-coordinate
running left to right) and vertically (Y-coordinate running bottom
to top).  Thus the points enclosed by the grid correspond to the
points of the unit square in the cartesian "X-Y" plane, the lower
left corner corresponding to the origin (0,0).  By extrapolation,
the full page corresponds to a larger rectangle in the plane.

     These coordinates serve to position labels as follows.
Before the \AffixLabels{...} command type label specifications:

  \SetLabels
   (<X-coordinate>*<Y-coordinate>) <first label> \\
   .
   .
   .
   (<X-coordinate>*<Y-coordinate>)  <last label> \\
  \endSetLabels

Each row specifies one label and is terminated by \\.  In each
row, the position indicator comes first; it is written as a
standard cartesian point except that the X- and Y- coordinates
are separated by * rather than a comma because TeX allows a
comma as decimal point. There are no dimension units to specify
as the unit is the grid itself.

     By default, this cartesian point specifies where the middle
of the baseline of the label will be located.  However if you precede
the point by \L [or \R] the left [or right] edge of the baseline will
be located there. Similarly you may also precede the point by \T, \E,
or \B to vertically align the top equator or bottom of the label box
at the specified point.  This gives nine standard positions of
the label with respect to the insertion point --- corresponding to
the eight principle points of the compas and the center

                     \L\T     \T      \R\T

                     \L\E     \E      \R\E

                     \L\B     \B      \R\B

But this neglects the default "baseline" level of TeX,
giving potentially three more positions

                     \L    <no tag>   \R

For text, the baseline level is often the preferred. Its relation to
the others is variable. It will often coincide with the bottom level,
as happens for "X".  But it is often distinct, as for "g", in which
case you have in all 12 distinct positions rather than 9.

     It is convenient to think of this specification of label
position as attaching the label by a thumb-tack to the coordinate
grid. There are up to twelve positions of the thumb-tack on the
label, while the position of the thumb-tack on the coordinate grid is
arbitrary.  Normally, one choses the position of the thumb-tack on
the label to be the one that is the closest to the item being
labeled.  There are good reasons for this "rule of thumb":

   (a)  It facilitates correct positioning at first try.

   (b)  If the scale of the figure must be altered after labels
have been affixed, the labels have a good chance of remaining well
positioned.

   (c)  The visible grid need not extend beyond the "bounding box"
for the figure, because the best preferred position is always
(at least almost) within the bounding box .

The second reason is particularly important. Indeed it often
happens that scale has to be altered after labelling begins, in
order to either provide space for the labels, or to adjust
proportions between the labels and the figure.  (The size of labels
is unaffected by scaling.)

     Here is an artificial but self-contained test which uses
TeX rules to make a graphics object.

TEST

    Do not skip this!

 \def\FrameIt#1{\hbox{\vrule$\vcenter {\hrule\kern3pt%
             \hbox {\kern3pt #1\kern3pt}%
               \kern3pt\hrule}$\relax\vrule}}

 \def\Caption#1#2{\FrameIt{%
       \vtop {\hsize=#1\relax \parindent=0pt
         \leftskip=0pt \rightskip=0pt plus15pt
         \parfillskip=0pt
         \lineskip=1pt\baselineskip=0pt
         #2}}}

 \def\FirstQuadrant{\hbox to 100pt{\vrule\vbox to 100pt{%
        \hbox to 100pt{\hfil}\vfil\hrule}\hss}}

  \SetLabels
    \R(.5*.2) $\zeta\,\cdot$\\
    (.9*-.10) $\xi$\\
    \R(-.03*.9) $\eta$\\
    \T(.5*.9) \Caption{70pt}{%
          \it The norm of
          $g(\xi+i\eta)$ is indicated on
          contours of this invisible surface.}\\
  \endSetLabels

  \AffixLabels{\FirstQuadrant}

  \end

  Note that the coordinates to use for labels are indicated on the
edges of the grid (when visible) corresponding to the conventional
x- and y- axes of the Cartesian plane. By default the grid is
1-by-1. However, by the command \Edges{100}, you can change this
to 100-by-100 and many users find this alternative most
convenient. Place the command \Edges{...} in your style file (or
header) since its effect is is global. Other possible edge values
are 10 and 1000.

  If you use the command \Edges{...} at all, do so with care.  For
if you accidentally delete an \Edges{...} command your labels will
abruptly be badly misplaced and may logically but mysteriously
generate "dimension too big" errors under TeX and "off page" errors
under your driver.  

  You can dictate the edgescale for an individual figure by giving
the scale in brackets immediately after \SetLabels.  Thus, to
import into an article using say \Edge{100} a figure labelled using
another edgescale, say the original 1-by-1 default, you can use
\SetLabels[1]...\endSetLabels.

GETTING IT DOWN PAT

     Complicated labeling deserves the same respect as
complicated mathematics.  Do not expect it to come out perfect the
first time!  What is needed in either case is a mechanism to
repeatedly typeset troublesome pieces.

     One mechanism is always available.  One does complicated
labelling in a separate "test" file involving just the figure being
labelled;  a texpert will know how to \dump TeX's current state as
a temporary format that restarts rapidly at each retry.  Usually,
one then pastes the completed labelled figure back into the main
TeX file, but, of course, one can also \input it as an auxiliary
file.

     If you do not have a TeXpert at handy, here is a first
approximation to an efficient setup. By deletions reduce a copy
of your article to just a few lines before and after the figure.
Now label the figure, and finally, copy and paste the labelled
figure to the original article. Then copy the next figure to label
into this testbed and repeat. The TeXpert can improve the  speed
at which TeX starts up, by compiling a format specifically for
your article; just one caution: best NOT include in the format
ephemeral details of setup like \Set<mydriver>ArtSpecials (from
boxedeps.tex because this reads  figure dimensions which you may
change during your work session.

     An improved mechanism to repeatedly typeset troublesome
pieces is now available on the Macintosh; it is called LinoTeX;
see the same ftp sources.  It could be set up on many types
of computer.

     Before using labelfig.tex to attach labels to a graphics
object inserted using boxedeps.tex or BoxedArt.tex, make it a
firm rule to carefully adjust the bounding box using the trimming
commands of these packages, and also at least tentatively scale
and position the object. Beware of changing the grid inadvertently
after the labels have been positioned.  For example, correcting
the bounding box of a PostScript graphics object can foul up the
labels by changing the coordinate grid to which the labels are
attached. This is particularly true for the trimming  commands of
boxedeps.tex and BoxedArt.tex. However, as noted already, change
of scale is much less disruptive, and modest adjustments should be
well tolerated.

     Sometimes the labels protrude so far from the bounding box
of a figure that the figure has to be repositioned.  Best do this
by ad hoc spacing, say using \hglue and \vglue; altering the
bounding box would create a vicious circle.

     Remember that you are responsible for preventing labels
from overlapping. You are responsible for all label typography
including size and style. A label is really just about anything
that can be put in a TeX box. Note that spaces at the beginning
and end of labels will normally be suppressed; if you really want
them you must protect them with TeX braces.

     This package temporarily sets the \mathsurround parameter
of TeX to zero  while the labels are being affixed. This is done
because nonzero \mathsurround space would influence the position
of left and right aligned labels; then, when a texpert or printer
modifies mathsurround, diagram labeling might be disastrously
altered. There is a small price to pay involving labels that are
formatted as caption boxes including mathematics: you  may want or
need to specify an explicit mathsurround space within the caption
box; it will not influence anything outside.

     Those hostile to the use of * as separator between
the X and Y coordinates of label insertion points, are free to
impose another using \SetXYSeparator{<the new separator>}.  
Americans may prefer "," to "*" since they never use a 
comma as a decimal point; on the other hand, * may be more visible.

APPENDIX (I)  MERGING labelfig.tex LABELS INTO AN EPSF GRAPHICS OBJECT.

     As promised in the introduction, here is a recipe useful for
publishers. It works at least on Macintosh and at least for vectorized
graphics and Adobe type1 fonts.  (There is surely a similar recipe for
PCs under MSWindows.)

 (a)  Use boxedeps.tex utility to integrate the figure given by the eps
file, "x.eps" say, with a visible frame around it.  See
\ShowDisplacementBoxes command in boxedeps.tex.  To get precise results
automatically it is important to use the \Trim... commands of
boxedeps.tex making the "DisplacementBox" neatly fit the figure.

 (b)  Use the TeX printer driver and LaserWriter (versions >= 8.1.1) to
export to an EPSF the DVI page containing the integrated, labelled
figure. You now have an EPS file  "xx.eps"  that contains too much, and at
the wrong scale, and at wrong position.

 (c)  Convert the EPSF to an Adode Illustrator format EPSF using
the shareware utility called epsConvert by Sam Weiss
1993-- (currently $25).

 (d)  In Illustrator (or a compatible program), group the labels and the
"DisplacementBox"; copy them to the clipboard and paste them into "x.ps".
This step requires that all the label fonts be "visible to the Macintosh.

 (e)  Translate and scale the pasted group consisting of the labels plus
the "DisplacementBox" so as to make the "DisplacementBox" the bounding
box of (labelless) figure represented by "x.eps".  At this point the
labels will be correctly placed on the figure "x.eps".

 (f)  Ungroup and delete the "DisplacementBox".  The result is the
desired single EPS file, "x+.eps" say, It contains the original figure
plus its labels.  

     Using grouping and ungrouping appropriately in "x+.eps", a
publisher's staff can very efficiently improve label positions etc.

APPENDIX II)  SOME EXOTIC APPLICATIONS

     The grid of labelfig.tex is analogous to a light-table in
classical page makeup with wax or latex glue.  In principle, you
can use it to compose any page from its indivisible parts.  This
even has some of the artisanal charm of classical paste-up
provided you have a fast screen preview to make the process
"interactive".

     In practice labelfig.tex is a tool for nonstandard jobs.
Here are a few going beyond the labelling already discussed.

(I)  GRAPHICS INTEGRATION.

     This is accomplished by treating the imported graphics
objects as labels.  The underlying graphics object is then
typically an empty  \vbox to <dimension>{\vfill} in a TeX
\midinsert...\endinsert construction.  A label line
might be of the form

   (.1*.1) \\

The exact form of the special command varies from driver to
driver.  However, in the case of encapsulated PostScript graphics
(EPSF norm), by relying on boxedeps.tex, one can have the
following standard syntax (independant of driver  (see
boxedeps.doc for details.
  
  (.1*.1) \BoxedEPSF{MyFigure scaled <scale in mils>}\\

This may be slow since it requires TeX to read the PostScript
file to read bounding box using many complex macros.  So you
may want to try

  (.1*.1) \EPSFSpecial{MyFigure}{<scale in mils>}\\

which is fast and driver independant, but it squashes the
bounding box, normally to its lower left corner.

     Similarly for graphics of the Macintosh PICT norm ---
using BoxedArt.tex (same sources) in place of boxedeps.tex.

     This approach to integration is to be recommended when
one is assembling a composite graphics object.

 (II)  COMMUTATIVE DIAGRAM ENHANCEMENT

     Commutative diagrams or arrays of mathematical objects
connected by arrows of various sorts are common in mathematics.
The mathematical objects require the use of TeX.  Recently TeX
acquired a good collection of arrows of all slopes --- that of
LamSTeX --- plus pwerful macros to build the diagrams.

     However, even the LamSTeX collection is often
inadequate; it lacks for example double shafted arrows, dotted
arrows and curved arrows. Fortunately it is possible to produce
such arrows on an individual basis using sophisticated graphics
programs such as Illustrator and AldusFreehand (both serving
the EPSF norm) or using Metafont (with its public domain norm).
Since the creation of each new arrow is a work of love, you
probably want to limit the number of arrows by using LamSTeX
for most arrows. The 40K commutative diagram module of LamSTeX
has been adapted to work with AmSTeX and a copy may be posted
with LabelFig and related files. Unfortunately no one has yet
offered a version that works with Plain TeX or LaTeX.

       Suffice it here to say that when the exotic arrow has
been somehow imported into TeX, labelfig.tex treats it as a
label that one affixes to the commutative diagram.  Two other
steps will be treated in separate notes, namely the matter of
extracting the dimension specifications for the arrow and the
construction of the arrow --- for these steps are far from
unique and often depend intimately on your computer environment. 
Notes for the Macintosh-Textures-Illustrator combination are
found in the file ExoticArrows.doc.

 (III) NESTING 

Ingenuity pays off in exploiting labelfig.tex. One can
mix graphics and typography quite freely.  labelfig.tex is good
for freeform or overlapping arrangements, while boxedeps.tex (or
BoxedArt.tex) is best for regimented non-overlapping
arrangements --- and the two can be combined.

     The default behavior of labelfig.tex is not ideal 
for nesting objects, because to prevent trouble for beginners
the register for labels is globally cleared when \AffixLabels
concludes.  But there are switches available

      \LabelsGlobal      \LabelsLocal

which change this.  To understand this, extend the above test 
by something like:

 \LabelsLocal

 \SetLabels
    (.5*.5) AAA\\
 \endSetLabels

 {
 \SetLabels
    (.5*.5) ZZZ\\
 \endSetLabels
   \AffixLabels{\FirstQuadrant}
 }

   \AffixLabels{\FirstQuadrant}

     There are however potential pitfalls.  Neither
labelfig.tex nor boxedeps.tex has been tested under extreme
conditions. Problems may occur if their procedures are
indiscriminately nested. For boxedeps.tex (not labelfig.tex)
there is a precise cause for worry, namely many of its
variables are "global", which means that TeX braces will not
provide the protection one might expect.

COMMAND SUMMARY FOR labelfig.tex

  Here [...] means optional (one or zero)
       [...]* means any number of such constructs

  \SetLabels
    [[<P>](<X><Sep><Y>) <label> \\]*
  \endSetLabels
  \ShowGrid  
  \AffixLabels{<the figure>}

   --- <P> is tack position, one of eleven or empty
              order irrelevant

                   \L\T      \T      \R\T

                   \L\E      \E      \R\E

                     \L               \R

                   \L\B      \B      \R\B

   --- (<X><Sep><Y>) insertion point;
  <Sep> is separator, = * by default;
  \SetXYSeparator{<Sep>} changes it.
   <X> and <Y> are real numbers

  --- <label> a label to attach 

  --- <the figure> the figure to label 

  \GlobalLabels (default)     
  \LocalLabels  setting for nested constructs.

 \Grids makes ALL grids appear; \HideGrid then makes just next disappear.
 \noGrids returns to default.  The commands are always global.

 \GridLineWidth{<dimension>} adjusts width of grid lines. Default is very
small, to give "hairline" effect. If your grid lines are missing try
setting \GridLineWidth{1pt}.

 \Edges#1 globally changes the edge size of all grids to the numerical 
value #1, which must be 1, 10, 100, or 1000.  The default is 1.

VERSION HISTORY.
 --- Jan 1993: \Edges#1 and [??] option after \SetLabels
 --- July 1992: \Grids, \noGrids, \HideGrid;
       Gridlines become hairlines; \GridLineWidth{<dimension>}.
 --- Oct 1991, Jan 1992: \SetXYSeparator{<Sep>},  \LabelsGlobal,
       \LabelsLocal.
 --- July 1991: first release

Address for bugs and other feedback:

        Raymond S\'eroul
        IREM and Lab. de Typographie Informatise
        Univ. Rene Descartes
        Strasbourg

    Tel 33-88-41-63-45
    Email:  A18645@FRCCSC21.BITNET

        Laurent Siebenmann
        Mathematique, Bat. 425,
        Univ de Paris-Sud,
        91405-Orsay,
        France

    Tel 33-1-6941-7949; 
    Email: lcs@topo.math.u-psud.fr

\def\scalefig#1{\epsfxsize #1\textwidth}

\newcommand {\Ebb}{{\mathbb{E}}}

\newtheorem{assumption}{Assumption}

\newtheorem{remark}{Remark}

\newtheorem{proposition}{Proposition}

\title{\LARGE {Training Beam Sequence Design for Millimeter-Wave MIMO Systems: A POMDP Framework}}

\author{
 Junyeong Seo, {\em Student~Member, IEEE}, Youngchul
Sung$^\dagger$\thanks{$^\dagger$Corresponding author}, {\em
Senior~Member, IEEE} \\ Gilwon Lee, and Donggun Kim, {\em Student~Members, IEEE} \\
\thanks{The authors are with Dept. of Electrical Engineering,  KAIST, Daejeon 305-701, South
Korea. E-mail:\{jyseo@, ysung@ee., gwlee@, and dg.kim@\}kaist.ac.kr.
This research was supported by Basic Science Research Program through the National Research Foundation of Korea (NRF) funded by the Ministry of Education (2013R1A1A2A10060852).}
}

\begin{document}

\maketitle

\begin{abstract}
In this paper,  adaptive training beam sequence design for
efficient channel estimation in large millimeter-wave (mmWave)
multiple-input multiple-output (MIMO) channels is considered. By
exploiting the sparsity in large mmWave MIMO channels and imposing
a Markovian random walk assumption on the movement of the receiver
and reflection clusters, the adaptive training beam sequence
design and channel estimation problem  is formulated as a
partially observable Markov decision process (POMDP) problem that
finds non-zero bins in a two-dimensional grid. Under the proposed
POMDP framework, optimal and suboptimal adaptive training beam sequence design
policies are derived. Furthermore, a very fast suboptimal greedy
algorithm is developed based on a newly proposed reduced
sufficient statistic to make the computational complexity of the
proposed algorithm low to a level for practical implementation.
Numerical results are provided to evaluate the performance of the
proposed training beam design method. Numerical results show that
the proposed training beam  sequence design algorithms yield good
performance.
\end{abstract}


\section{Introduction}

The use of the mmWave frequency band allows wide
bandwith for high data rates required for future wireless
networks. However, mmWave signals experience severe large-scale
pathloss compared to lower frequency band signals, which has been
a hurdle to using the mmWave band for commercial wireless access
networks so far.  Recently, active research is going on to use the
mmWave band for cellular systems by exploiting advanced hardware
and software processing power \cite{Venkateswaran&Veen:10SP,
Ayach&HeathEtAl:12ICC, Alkhateeb&Ayach&Leus&Heath:13ITA,
Alkhateeb&Ayach&Leus&Heath:14arxiv}.  One of the major techniques
to compensate for the large pathloss in the mmWave band is highly
directional beamforming based on large antenna arrays
\cite{Ayach&HeathEtAl:12ICC, Alkhateeb&Ayach&Leus&Heath:13ITA,
Alkhateeb&Ayach&Leus&Heath:14arxiv} which can be implemented in
small sizes in the mmWave band \cite{Doan&Emami&Sobel:04ComMag}.
Typically such beamforming requires channel state information
(CSI) at the transmitter and the receiver, but the CSI is
difficult to acquire in the mmWave band due to the propagation
directivity and the low signal-to-noise ratio (SNR) before
beamforming due to large pathloss.  Thus, efficient and accurate
channel estimation is one of the key requirements for the success of
large mmWave MIMO systems
\cite{Ayach&HeathEtAl:12ICC, Alkhateeb&Ayach&Leus&Heath:13ITA,
Alkhateeb&Ayach&Leus&Heath:14arxiv} .

Conventional training-based MIMO channel estimation methods
typically assume rich scattering environments or the
knowledge\footnote{The knowledge of the channel covariance matrix
in the rank-deficient channel case implies that the propagation
ray directions are known to the transmitter {\it a priori}.} of
the channel covariance matrix in the rank-deficient channel case
\cite{Telatar:00Euro, Hassibi&Hochwald:03IT, Kotecha&Sayeed:04SP,
Choi&Love&Bidigare:14JSTSP,
Noh&Zol&Sung&Love:14JSTSP,So&Kim&Lee&Sung:14SPL}. However, such
assumption may not be valid in large mmWave MIMO systems due to
the high propagation directivity and the narrow beam width
associated with large antenna arrays
\cite{Bajwa&Haupt&Sayeed&Nowak:10IEEE}. Typical mmWave channels
with large antenna arrays can be modeled as sparse MIMO channels
with the virtual channel representation and the directions  of ray
clusters are unknown to the transmitter and the receiver beforehand
\cite{Taubock&Hlawatsch:08ESP, Bajwa&Haupt&Sayeed&Nowak:10IEEE,
Alkhateeb&Ayach&Leus&Heath:13ITA,
Alkhateeb&Ayach&Leus&Heath:14arxiv} (see Fig.
\ref{fig:channel_model}). Thus, conventionally designed training
signals and channel estimation methods aiming at the lower
frequency band are less efficient, and the training signal design
and channel estimation are  more challenging in the mmWave case.
One way to identify a sparse channel is to transmit all beam
directions sequentially in time and pick the direction of the
largest signal magnitude \cite{Wang&Lan&PyoEtAl:09JSAC,
Hur&Kim&LoveEtAl:13COM}. However, such a method is not efficient
when the number of all possible directions to search is large as
in the large mmWave MIMO case.  To tackle the challenge of sparse
MIMO channel estimation, algorithms based on compressed sensing (CS) theory have recently
been developed \cite{Taubock&Hlawatsch:08ESP,
Bajwa&Haupt&Sayeed&Nowak:10IEEE, Alkhateeb&Ayach&Leus&Heath:13ITA,
Alkhateeb&Ayach&Leus&Heath:14arxiv}. In \cite{
Bajwa&Haupt&Sayeed&Nowak:10IEEE}, the problem of channel
estimation in large mmWave MIMO  was formulated
 to capture the sparse nature of the channel and CS tools were applied to analyze the sparse channel estimation performance.
In particular, in \cite{Alkhateeb&Ayach&Leus&Heath:13ITA,
Alkhateeb&Ayach&Leus&Heath:14arxiv}, Alkhateeb {\it et al.}
proposed a channel estimation and training beam design method for
large mmWave MIMO systems based on adaptive CS. In their method,
the channel estimation is performed over multiple blocks under the
assumption that the channel does not vary over the considered
multiple blocks. Each block consists of multiple training beam
symbol times so that the sparse recovery is feasible at each
block, and the next training beam is adaptively designed based on
the previous block observation result by using a space bisection
approach which is a reasonable choice to search a propagation ray cluster in the space.

In this paper, exploiting the channel dynamic, we propose a
different training beam design and channel estimation paradigm for
 sparse large mmWave MIMO systems based on a decision-theoretical
framework. We  consider a typical time-duplexed training
structure, where one slot consists of a block of pilot symbol
times and the following symbol times for data transmission
\cite{Tong&Sadler&Dong:04SPM}. We assume that still a highly
directional narrow pilot beam should be transmitted at each
training symbol time to compensate for large pathloss  and obtain
a reasonable quality channel gain estimate in the mmWave band.
Then, the training beam design problem reduces to the problem of
choosing the directions of the pilot beams in the space to find
the actual propagation paths generated by line-of-sight and/or
reflection clusters, as shown in Fig. \ref{fig:channel_model}.
Specifically, with the virtual channel representation
\cite{Bajwa&Haupt&Sayeed&Nowak:10IEEE}, the sparse channel
estimation reduces to finding the locations and values of the
non-zero valued bins in a two-dimensional grid. Our main idea is
to impose a Markovian random walk structure on the movement of the
mobile station and the reflection clusters\footnote{This
assumption seems reasonable when we consider the physics of the
mobile station or the reflection clusters. For the example of a
pedestrian user with a certain speed of walking, the propagation
ray directions at the next slot changes from the current
directions and this uncertain change can be captured as a random
walk.}. In the proposed scheme, a set of pilot beams is
transmitted at each slot and the set of pilot beams at the current
slot is adaptively and optimally determined based on the
observation over all the previous slots to maximize a properly
defined reward accumulated over a given communication period.
Since we cannot observe all possible ray directions in the space
at each slot and the identification of a path can be wrong, the
pilot beam design under this formulation reduces to a {\em
partially observable Markov decision process (POMDP)}
\cite{PutermanBook}, and the theory of POMDP can be applied to the
pilot beam design problem for large mmWave MIMO. Under the
proposed POMDP formulation optimal and suboptimal greedy
strategies for training beam sequence design for sparse large MIMO
channel estimation are derived.  However, the direct application
of standard POMDP theory yields an intractable number of states
and unfeasible complexity. Thus, exploiting the specific structure
of the mobile communication channel and deriving a new
reduced-size sufficient statistic for the decision process, we
develop
 a greedy algorithm with significantly reduced complexity so that the proposed algorithm can practically be operated. Numerical results
 show that the proposed low complexity suboptimal algorithm yields comparable performance
relative to the optimal algorithm, and the proposed
training beam design algorithms efficiently estimate and tract sparse MIMO
channels. (A preliminary version of this work was submitted to ICC
2015 \cite{Seo&Sung&Lee&Kim:15ICCsub}.)

 {\it Notations and Organization} ~~~ We will use standard notational conventions in this paper.
Vectors and matrices are written in boldface with matrices in capitals.
For a matrix $\Abf$, $\Abf^T$, $\Abf^H$, $[\Abf]_{ij}$, $\Abf(:,k)$, and $\mbox{tr}(\Abf)$ indicate the transpose,  conjugate  transpose, the element of
the $i$-th row and the $j$-th column,  the $k$-th column, and trace  of $\Abf$, respectively.
$\Ibf_n$ stands for the identity matrix
of size $n$ and ${\mathbf{1}}_n$ stands for the matrix of size $n$ whose entries are all one. $\Abf \otimes \Bbf$ is the Kronecker product of $\Abf$ and $\Bbf$. The notation $\xbf\sim
\Cc(\mubf,\Sigmabf)$ means that $\xbf$ is
complex Gaussian distributed with mean vector $\mubf$ and
covariance matrix $\Sigmabf$.
$\Ebb\{\cdot\}$ denotes the
expectation. $|\abf|$ and $\|\abf\|_0$ denote the number of
elements and the number of non-zero elements of $\abf$, respectively. $\iota:=\sqrt{-1}$.

This paper is organized as follows. In Section \ref{sec:systemmodel}, the system model is explained.
In Section \ref{sec:POMDPformulation}, the optimal
training beam sequence design problem is formulated as a POMDP problem. In Section \ref{sec:strategy},
 optimal and suboptimal strategies are presented and a greedy
pilot beam design algorithm with low complexity is derived. Numerical
results are provided in Section \ref{sec:NumericalResult},
followed by conclusions in Section \ref{sec:conclusion}.

\section{System Model}
\label{sec:systemmodel}

\subsection{Sparse Channel Modeling in Large mmWave MIMO Systems}
\label{subsec:sparsChannelModel}

We consider a  mmWave MIMO system, where a transmitter equipped
with a uniform linear array (ULA) of $N_t$ antennas communicates
to a receiver equipped with a ULA of $N_r$ antennas. The received
signal at symbol time $n$  is then given by
\begin{equation}
\ybf_n= \Hbf_n \xbf_n + \nbf_n, ~~ n=1,2,\cdots,
\end{equation}
where $\Hbf_n$ is the $N_r\times N_t$ MIMO channel matrix at time
$n$, $\xbf_n$ is the $N_t \times 1$ transmit symbol vector at time
$n$ with a power constraint $\mbox{tr}(\Ebb\{ \xbf_n \xbf_n^H \})
\le P_t$, and $\nbf_n$ is the $N_r \times 1$ Gaussian noise vector
at time $n$ from $\Cc\Nc({\bf 0}, \sigma_N^2 \Ibf_{N_r})$.

A physical multipath channel accurately modeling $\Hbf_n$ is given by  \cite{Sayeed:02SP,Sayeed&Raghavan:07JSTSP}
\begin{equation}\label{eq:physical_channelmodel}
\Hbf_n = \sqrt{N_tN_r}\sum_{l=1}^L
\alpha_{n,l}\abf_{RX}(\theta_{n,l}^r)
\abf_{TX}^H(\theta_{n,l}^t),
\end{equation}
where $\alpha_{n,l} \sim \Cc\Nc(0,\xi^2)$ is the complex gain
of the $l$-th path at time $n$, and $\theta_{n,l}^r$ and
$\theta_{n,l}^t$ are the angle-of-arrival (AoA) and
angle-of-departure (AoD) normalized directions of the $l$-th
path at time $n$ for the receiver and the transmitter,
respectively. Here, the normalized direction $\theta$ is related
to the physical angle $\phi \in [-\pi/2,\pi/2]$ as
\begin{equation}
\theta = \frac{d\sin(\phi)}{\lambda},
\end{equation}
where $d$ and $\lambda$ are the spacing between two adjacent antenna elements and the signal wavelength, respectively.
We assume $\frac{d}{\lambda} = \frac{1}{2}$ and thus, the range of $\theta$ is $[-\frac{1}{2}, \frac{1}{2}]$.
In \eqref{eq:physical_channelmodel}, $\abf_{RX}(\theta^r)$ and $\abf_{TX}(\theta^t)$
are the receiver response and the transmitter steering vector in normalized
directions $\theta^r$ and $\theta^t$, respectively, which are defined as
\cite{Sayeed&Raghavan:07JSTSP}
\begin{align}
\abf_{RX}(\theta^r) &= \frac{1}{\sqrt{N_r}}
[1,e^{-\iota 2 \pi \theta^r},\cdots,e^{-\iota(N_r -1) 2 \pi  \theta^r}]^T, \\
\abf_{TX}(\theta^t) &= \frac{1}{\sqrt{N_t}}
[1,e^{-\iota 2 \pi \theta^t},\cdots,e^{-\iota(N_t -1) 2 \pi  \theta^t}]^T.
\end{align}
Note that $||\abf_{RX}(\theta^r)||=||\abf_{TX}(\theta^t)||=1$.
Neglecting the angle quantization error, we can map the physical
MIMO channel matrix $\Hbf_n$  to a virtual channel matrix (VCM)
$\widetilde{\Hbf}_{n}$  through the following relationship
\cite{Sayeed:02SP}
\begin{equation}
\Hbf_n = \Abf_R \widetilde{\Hbf}_{n} \Abf_T^H,
\end{equation}
where $\Abf_R =
[\abf_{RX}(\tilde{\theta}^r_{1}),\cdots,\abf_{RX}(\tilde{\theta}^r_{N_r})]$,
$\tilde{\theta}^r_{i}=-\frac{1}{2} + \frac{i-1}{N_r}$ for
$i=1,\cdots,N_r$, and $\Abf_T =
[\abf_{TX}(\tilde{\theta}^t_{1}),\cdots,\abf_{TX}(\tilde{\theta}^t_{N_t})]$,
$\tilde{\theta}^t_{j}=-\frac{1}{2} + \frac{j-1}{N_t}$ for
$j=1,\cdots,N_t$. (From here on, we will neglect the angle
quantization error.) The $(i,j)$-th element of $\widetilde{\Hbf}_{n}$
represents the channel gain scaled by $\sqrt{N_tN_r}$ when the virtual angles seen by the
receiver and the transmitter are $\tilde{\theta}^r_{i}$ and
$\tilde{\theta}^t_{j}$, respectively. For the element of $\widetilde{\Hbf}_{n}$ corresponding to the $l$-th propagation path, the element value
is given by $\sqrt{N_tN_r}\alpha_{n,l}$.
The physical channel model
\eqref{eq:physical_channelmodel} induces a sparse property to the
virtual channel representation,  given by
\begin{equation}
\sum_{j=1}^{N_t} ||\widetilde{\Hbf}_{n}(:,j)||_0 = L.
\end{equation}
That is, there are only $L ~(\ll N_tN_r)$ non-zero elements among
the $N_t N_r$ entries of the VCM $\widetilde{\Hbf}_{n}$.
\begin{figure}[t]
\begin{psfrags}
        \psfrag{lambda/2}[c]{\small $d$} %
        \psfrag{v1}[l]{\small $\abf_{RX}(\tilde{\theta}^r_{1})$} %
        \psfrag{v2}[l]{\small $\cdots$} %
        \psfrag{v3}[c]{\small $\abf_{RX}(\tilde{\theta}^r_{N_r})$} %
        \psfrag{u1}[c]{\small $\abf_{TX}(\tilde{\theta}^t_{1})^H$} %
        \psfrag{u2}[c]{\small $\vdots$} %
        \psfrag{u3}[c]{\small } %
        \psfrag{uNt}[c]{\small $\abf_{TX}(\tilde{\theta}^t_{N_t})^H$} %
        \psfrag{H=}[c]{ $\Hbf_n=$} %
    \centerline{ \scalefig{0.82} \epsfbox{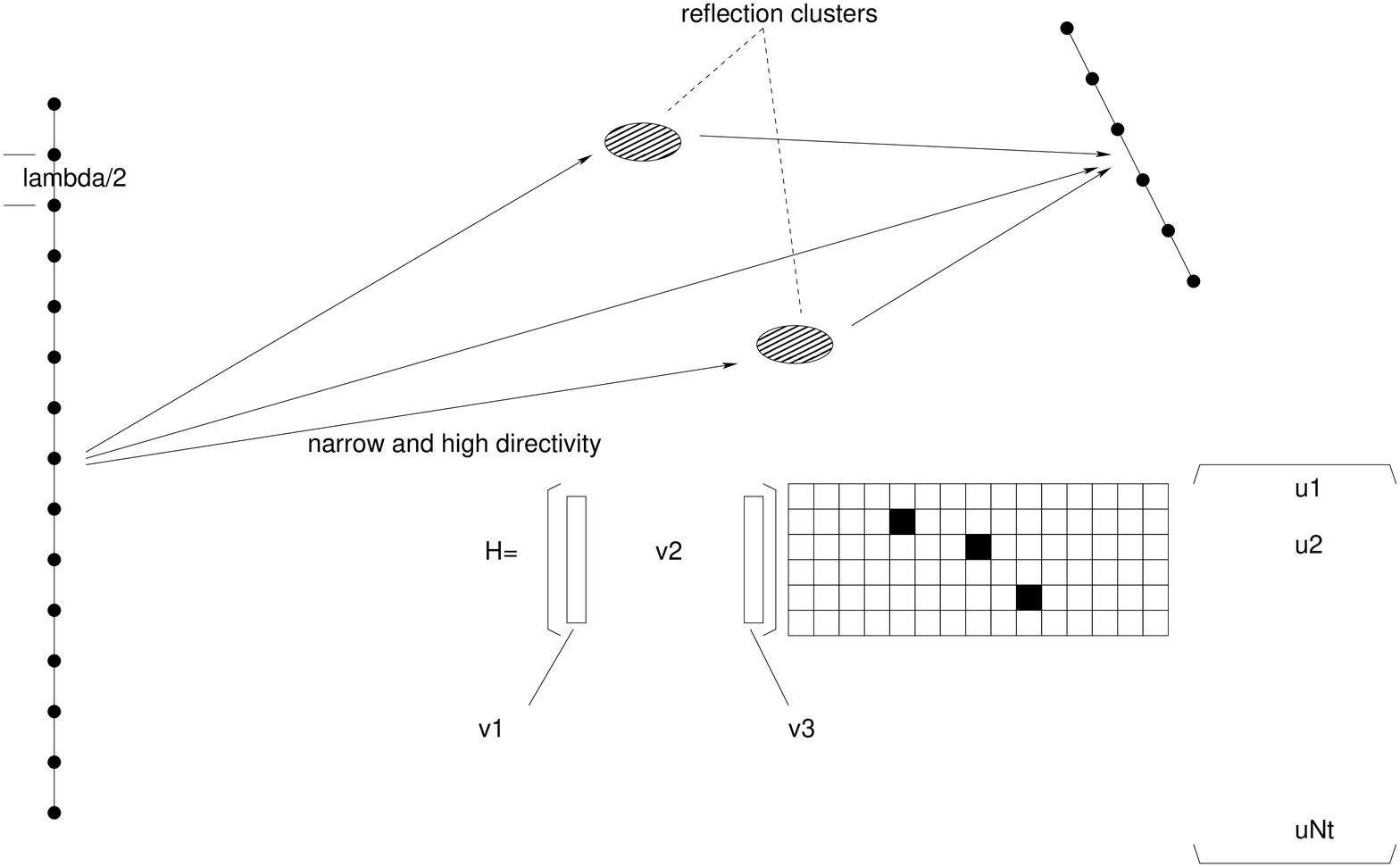} }
    \caption{The considered  mmWave channel model}
    \label{fig:channel_model}
\end{psfrags}
\end{figure}
Therefore, for large $N_t$ and $N_r$, it is generally difficult to
find the locations of the $L$ non-zero elements in
$\widetilde{\Hbf}_{n}$ without an exhaustive search.

For the sake of simplicity, with the cellular downlink in mind, we
assume that $N_r \ll N_t$ and the receiver has  $N_r$ monolithic
microwave integrated circuit  (MMIC)  RF chains so that the
receiver can implement the filter bank $\Abf_R^H$. Then, the
receiver-filtered signal at the filter bank output is given by
\begin{equation}\label{eq:filterbank_received}
\ybf_n' := \Abf_R^H\ybf_n= \widetilde{\Hbf}_{n}\Abf_T^H\xbf_n + \nbf_n',
\end{equation}
where $\nbf_n' = \Abf_R^H\nbf_n$.

Now, since the receiver checks all the possible (quantized) ray
directions, the remaining problem is to design the transmit
training beam sequence $\{ \xbf_n| n \in T_\Pc\}$ for estimating
the sparse channel, where $T_\Pc$ is the set of symbol times
allocated to training signal transmission.

\subsection{The Considered Dynamic Channel Model}

To design an efficient training beam sequence for estimating the
sparse channel presented in Section
\ref{subsec:sparsChannelModel}, we exploit the channel dynamic and
model the channel dynamic by imposing a block Markovian structure
on the VCM $\widetilde{\Hbf}_n$. That is, $\widetilde{\Hbf}_n$ is
constant over one slot consisting of $M_s$ symbols. Let us denote
the VCM at  slot $k$ by $\widetilde{\Hbf}_{(k)}$. The VCM
$\widetilde{\Hbf}_{(k)}$ at slot $k$
 changes to $\widetilde{\Hbf}_{(k+1)}$ at  slot $k+1$ in a
Markovian manner. Here, the conventional block Gauss-Markov
process or state-space channel model widely used in time-varying
channel estimation \cite{Choi&Love&Bidigare:14JSTSP, Noh&Zol&Sung&Love:14JSTSP} is
inappropriate to model the sparse mmWave MIMO channel, and the sparsity
of the mmWave MIMO channel should be captured in the Markov model.
 Focusing on the dynamic of the locations of the non-zero elements of the VCM rather than the values\footnote{The value will be obtained with reasonable quality once the
correct direction is hit by the pilot beam with high power.} and
considering that the non-zero elements in $\widetilde{\Hbf}_{(k)}$
are associated with the line-of-sight (LOS) and reflection
clusters,  we assume the following model:

\vspace{0.5em}

\begin{assumption} \label{ass:lpaths_indep}
Each of the $L$ paths (or the locations of the non-zero elements)
in $\widetilde{\Hbf}_{(k)}$ moves from the current column location
to another column location in $\widetilde{\Hbf}_{(k+1)}$ in a
Markovian manner with a transition probability, and the transition
probability does not change over the considered period of time of
total $T$ slots. Furthermore, the transition of each path is
independent.
\end{assumption}

\vspace{0.5em}

Here, we do not consider the row-wise transition of the non-zero
bins  in the VCM because we assume that the receiver has a filter
bank that checks all AoA directions in parallel at each symbol
time. The rationale for the above assumption is that each
propagation path is generated by either the LOS or a reflection
cluster and the physical movement of the LOS or a reflection
cluster can be modelled as a random walk in space. This random
walk translates into each nonzero bin's random walk in the VCM.
Thus, in the proposed model, the column location of the $l$-th
path follows a random walk on $\{1,2,\cdots,N_t\}$ and the path
gain of the $l$-th path is given by the sequence
$\{\alpha_{(1),l},\alpha_{(2),l},\cdots,$
$\alpha_{(k),l},\cdots\}$, where $\alpha_{(k),l}$ is the complex
gain of the $l$-th path at slot $k$.

\subsubsection{The Case of a Single Path}
\label{sec:singlepathcase}

First, consider the single path case, i.e., $L =1$. In this case,
the number $N$ of states  is $N_t$ because we have $N_t$ columns
in the VCM.  Thus,  the set $\mathcal{S}$ of all possible states
is
 given by
\[
\Sc = \{1,2,\cdots, N_t\},
\]
 where  state $i$  denotes
 the state that the path is located in the $i$-th column of the VCM.
With the set $\Sc$ of states defined, the $(i,j)$-th element of
the $N \times N$ state transition probability matrix $\Pbf$ is
defined as
 \begin{equation}\label{eq:transitionprobability}
    p_{ij} = \text{Pr}\{ S_{k+1} = j | S_k = i\}, ~~ i,j \in
    \mathcal{S},
\end{equation}
where $S_k$ and $S_{k+1}$ denote the states of slots $k$ and
$k+1$, respectively.
 The transition probability matrix $\Pbf$ captures the
characteristics of the movement behavior of the path and hence it should
 be designed carefully by considering the physics of the receiver
and reflection cluster movement. For example, vehicular channels or
pedestrian channels with certain speeds will have different the values in $\Pbf$.
Intuitively, it is reasonable to design
$\Pbf$ so that  the transition probability from the $i$-th column  to
the $j$-th column  monotonically decreases as $|i-j|$ increases. That is, it is more likely to shift to a nearby column with continuous movement.
 Ignoring the possibility of the path's movement
with a large AoD change per slot, we can model the transition
probability matrix as a banded matrix. In this case, for example,
we can consider the transition probability matrix with exponential
decay  given by
 \begin{equation}\label{eq:bandedstructure}
    \Pbf_{\beta}^{B} = \left[\begin{array}{cccccccccccc}
                           & &   & \ddots &  &  &  &  & &  &\\
                            0 & \alpha \beta^B  & \cdots  & \alpha \beta & \alpha & \alpha \beta & \cdots & \alpha\beta^B &0 & \cdots &\\
                          \cdots& 0 & \alpha \beta^B  & \cdots  & \alpha \beta & \alpha & \alpha \beta & \cdots & \alpha\beta^B &0 & \cdots\\
                         & \cdots& 0 & \alpha \beta^B  & \cdots  & \alpha \beta & \alpha & \alpha \beta & \cdots & \alpha\beta^B &0 \\
                          & &   &  &  &  &  & \ddots &  & &\\
                              \end{array}
    \right],
 \end{equation}
 where $B$ is the bandwidth of the matrix and $\beta \in [0,1)$ is the exponential decaying factor.
  (The value of $\alpha$ and the corner of  $\Pbf_{\beta}^{B}$ should be obtained properly so that the sum of each row  is one.) When we want to model the random path appearance from an arbitrary direction, we can use
\begin{equation}\label{eq:bandedstructure_random}
    \Pbf_{\beta,\lambda}^{B} = (1-\lambda)  \Pbf_{\beta}^{B} + \lambda \frac{1}{N_t}  {\mathbf{1}}_{N_t}
\end{equation}
where ${\mathbf{1}}_{N_t}$ is the $N_t \times N_t$ matrix whose entries are all
one and $\lambda \in [0,1]$. Note that
the proposed model can capture static channels by setting $\Pbf=\Ibf$.  Fig. \ref{fig:markov chain}
shows the Markov chain in the  single path case.

  \begin{figure}[t]
\begin{psfrags}
        \psfrag{cdot}[c]{\small $\cdots $} %
        \psfrag{a1}[c]{\small $1$}
        \psfrag{a2}[c]{\small $2$}
        \psfrag{a3}[c]{\small $N_t$}
        \psfrag{p11}[c]{\small $p_{11}$}
        \psfrag{p12}[c]{\small $p_{12}$}
        \psfrag{p22}[c]{\small $p_{22}$}
        \psfrag{p21}[c]{\small $p_{21}$}
        \psfrag{p13}[c]{\small $p_{1 N_t}$}
        \psfrag{p23}[c]{\small $p_{2 N_t}$}
        \psfrag{p32}[c]{\small $p_{N_t 2}$}
        \psfrag{p31}[c]{\small $p_{N_t 1}$}
        \psfrag{p33}[c]{\small $p_{N_t N_t}$}
    \centerline{ \scalefig{0.6} \epsfbox{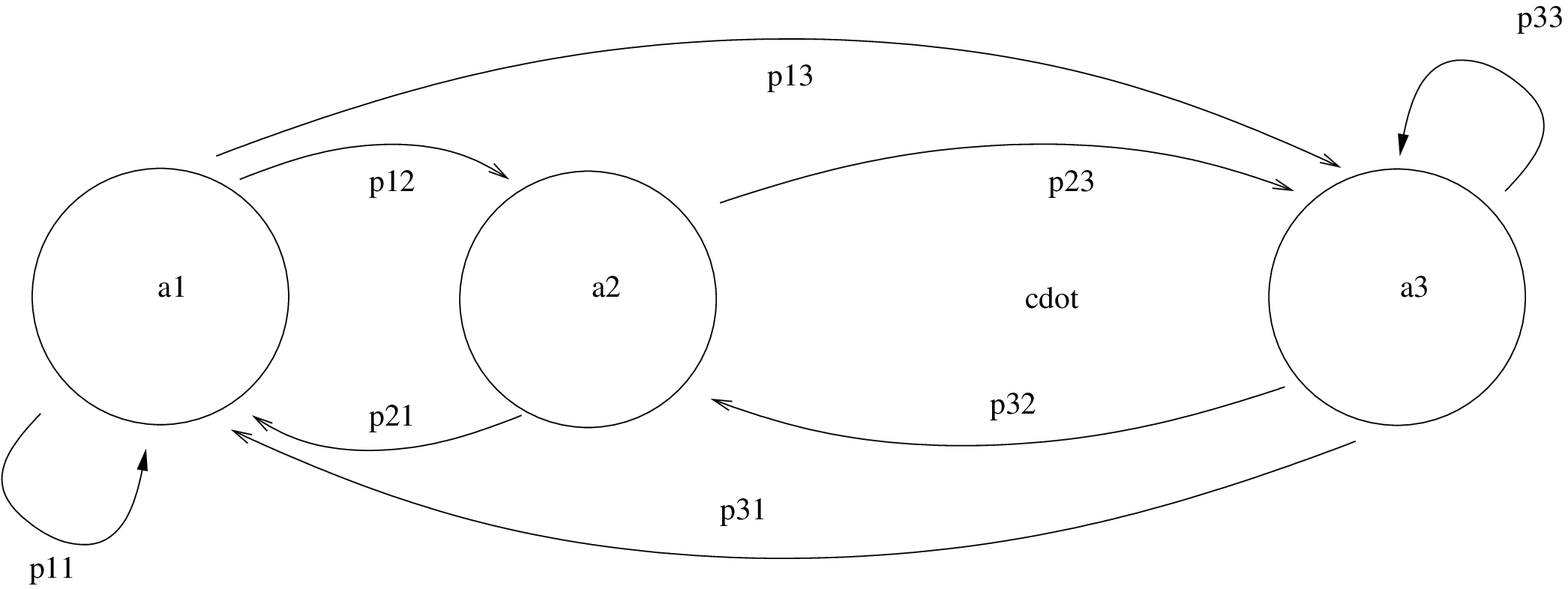} }
    \caption{The Markov chain model for the channel dynamic: The single path case}
    \label{fig:markov chain}
\end{psfrags}
\end{figure}

\subsubsection{The Case of Multiple Paths}
\label{sec:multipathscase}

Now consider the multi-path case, i.e., $L \ge 2$. We allow
multiple paths to merge on and diverge from a column of the VCM.
In this case, the set $\Sc$ of all possible states is given by
\begin{equation}
\Sc=\{(i_1,i_2,\cdots,i_L), ~i_1,i_2,\cdots,i_L =1, 2,\cdots,N_t\},
\end{equation}
where state $(i_1,\cdots,i_L)$ denotes that the $l$-th path is
located at the $i_l$-th column of the VCM for $l=1,\cdots,
L$, and the cardinality $N$ of $\Sc$ is $N_t^L$. Under the
assumption of independent path movement  the state transition
probability in the $L$-path case is given by
 \begin{align}
   & \text{Pr}\{ S_{k+1} = (j_1,\cdots,j_L) | S_k = (i_1,\cdots,i_L)\} \nonumber \\
   &~~= p_{i_1j_1}p_{i_2j_2} \times \cdots \times p_{i_L j_L}, \label{eq:transitionprobabilityL}
\end{align}
where $p_{ij}$ denotes the transition probability that a path
moves from the $i$-th column to the $j$-th column of the VCM at
the next slot, defined in \eqref{eq:transitionprobability}.
 States $(i_1,\cdots,i_L)$, $i_1,\cdots,i_L=1,\cdots,N_t$, can be enumerated as states $\sbf^{(i)}$, $i=1,2,\cdots, N=N_t^L$. Thus, $\Sc$ can also be expressed for notational simplicity as
\begin{equation} \label{eq:stateRepEnum}
\Sc=\{\sbf^{(1)},\sbf^{(2)},\cdots,\sbf^{(N)}\}.
\end{equation}

\begin{figure}[t]
\begin{psfrags}
        \psfrag{mp5}[c]{\small $-0.5$} %
        \psfrag{p5}[c]{\small $0.5$} %
        \psfrag{al1}[c]{\small $\alpha_{1,1}$} %
        \psfrag{al2}[c]{\small $\alpha_{1,2}$} %
        \psfrag{al3}[c]{\small $\alpha_{1,3}$} %
        \psfrag{aoa}[c]{\small $\{\tilde{\theta}^r_{m}\}$ (AoA)} %
        \psfrag{aod}[c]{\small $\{\tilde{\theta}^t_{n}\}$ (AoD)} %
        \psfrag{alk1}[c]{\small $\alpha_{2,1}$} %
        \psfrag{alk2}[c]{\small $\alpha_{2,2}$} %
        \psfrag{alk3}[c]{\small $\alpha_{2,3}$} %
        \psfrag{t1}[c]{\small  slot $k$} %
        \psfrag{t2}[c]{\small slot $k+1$} %
    \centerline{ \scalefig{0.9} \epsfbox{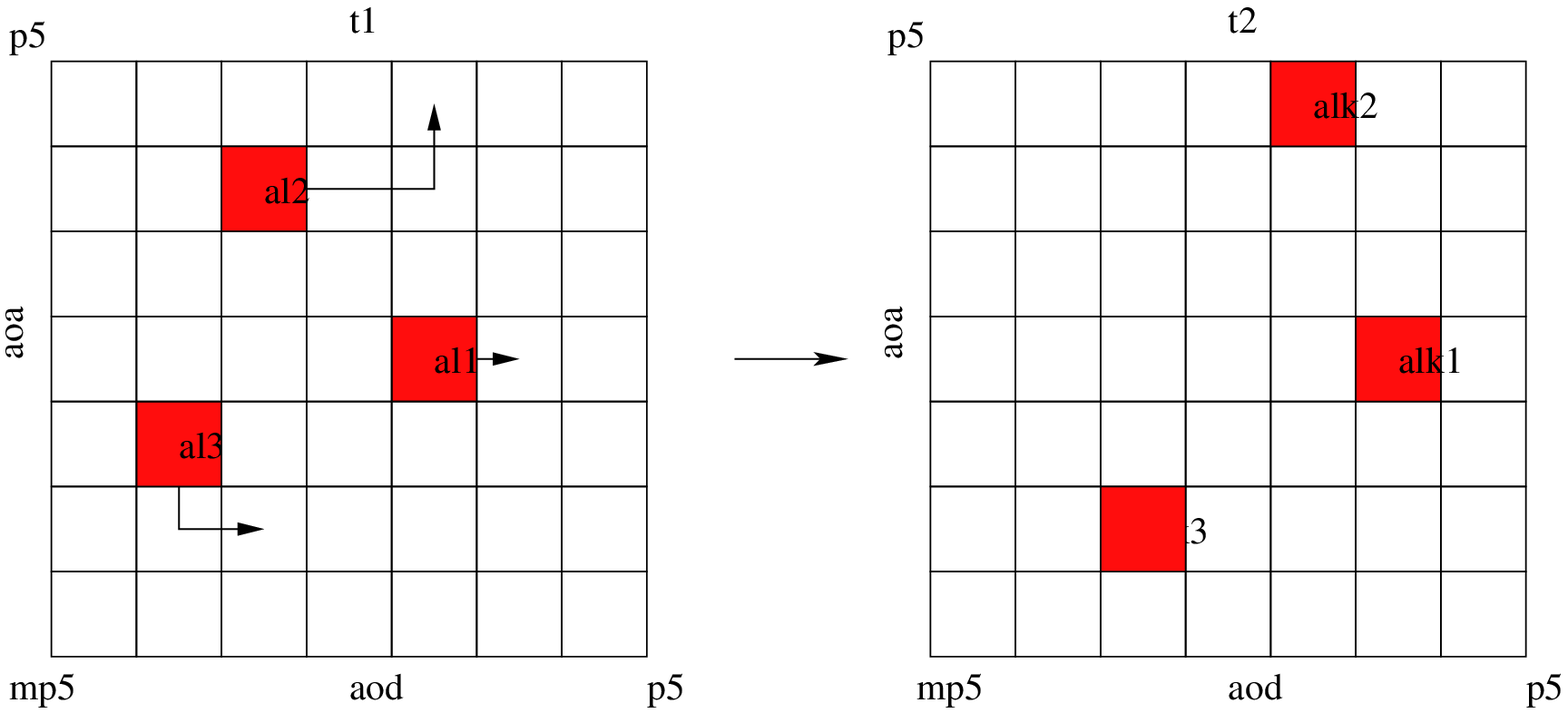} }
    \caption{An illustration of the transition of each path ($L=3$
    and $N_t=N_r=7$)}
    \label{fig:aod_bins}
\end{psfrags}
\end{figure}

Fig. \ref{fig:aod_bins} shows an example of the transition of each
path when $L=3$ and $N_t=N_r=7$. In this example, the AoD
directions of the $1$st, $2$nd, and $3$rd path move from
$\tilde{\theta}^t_{5}$ to $\tilde{\theta}^t_{6}$, from
$\tilde{\theta}^t_{3}$ to $\tilde{\theta}^t_{5}$, and from
$\tilde{\theta}^t_{2}$ to $\tilde{\theta}^t_{3}$, respectively.
 The probability of this  transition  is
$p_{56} \times p_{35} \times p_{23}$ under the assumption of
independent movement of each path.

\subsection{Channel Sensing with Pilot Beams}

We here explain the transmission structure. We consider a typical
time-duplexed training structure, where one slot of size $M_s$
slots consists of a block of $M_p$ $(L \le M_p\ll N_t)$ pilot
symbol times and the remaining symbol times of the slot are used
for data transmission \cite{Tong&Sadler&Dong:04SPM}. We assume
that the transmitter  picks one column of $\Abf_T$ as the pilot
beam at each pilot symbol time and search one column of
$\widetilde{\Hbf}_{(k)}$ at each pilot symbol time, i.e., $\xbf_n
\in \{\sqrt{P_t}\abf_{TX}(\tilde{\theta}^t_{1})$, $\cdots$,
$\sqrt{P_t}\abf_{TX}(\tilde{\theta}^t_{N_t}) \}$ for the training
time.
  Hence, $M_p$ columns of $\Abf_T$ are selected as the $M_p$ pilot
beams in a slot.  The reason for this assumption is that we
assume that the pathloss is severely large in the mmWave band and
thus the pilot beam as well as the data-transmitting beam should
be highly directional to compensate for the large pathloss and obtain a
channel gain estimate of reasonable quality, unless $P_t$ is
extremely high.  (The relaxation of this assumption will be
discussed in Section \ref{sec:extensions}.)
When $\sqrt{P_t}\abf_{TX}(\tilde{\theta}^t_{i_m})$ is the pilot beam at the $m$-th symbol time of
slot $k$, from
\eqref{eq:filterbank_received}, the receiver filter-bank output is given by
\begin{equation}\label{eq:receivedsig}
\ybf_{(k)}'[m] = \sqrt{P_t}\widetilde{\Hbf}_{(k)}(:,i_m) + \nbf_{(k)}'[m],
\end{equation}
where $\ybf_{(k)}'[m]$ and $\nbf_{(k)}'[m]$ denote the signal and
the noise at the receiver filter-bank output at the $m$-th pilot
symbol time of slot $k$, and $\widetilde{\Hbf}_{(k)}(:,i_m)$ is
the $i_m$-th column of $\widetilde{\Hbf}_{(k)}$.

During the training period, the receiver senses and estimates the
$M_p$ columns of the VCM corresponding to the $M_p$ pilot beams.
Then, the receiver feedbacks the sensing results and estimated
channel gains to the transmitter. Finally, the transmitter sends
data and adapts the pilot beams for the next slot based on the
information from the receiver. At the next slot, the process is
repeated.  The whole process of training and data transmission is
depicted in Fig. \ref{fig:pilots_model}.
\begin{figure}[t]
\begin{psfrags}
        \psfrag{mp}[c]{\small $M_p $} %
        \psfrag{st}[c]{\small State transition}
        \psfrag{k}[c]{\small  slot $k$}
        \psfrag{k1}[c]{\small slot $k+1$}
        \psfrag{cd}[c]{\small $\cdots$}
        \psfrag{cs}[c]{\small Channel sensing}
        \psfrag{ce}[c]{\small and estimation}
        \psfrag{rx}[c]{\small RX}
        \psfrag{tx}[c]{\small TX}
        \psfrag{hb}[c]{\small $\Hbf_k$}
        \psfrag{fb}[c]{\small Feedback information}
        \psfrag{dp}[c]{\small Data transmission}
        \psfrag{td}[c]{\small A sequence of pilot beams
        for $M_p$ time slots}
        \psfrag{a1}[c]{\small $\abf_{TX}(\tilde{\theta}^t_{i_1})$}
        \psfrag{a2}[c]{\small $\abf_{TX}(\tilde{\theta}^t_{i_2})$}
        \psfrag{a4}[c]{\small $\abf_{TX}(\tilde{\theta}^t_{i_{M_p}})$}
    \centerline{ \scalefig{0.9} \epsfbox{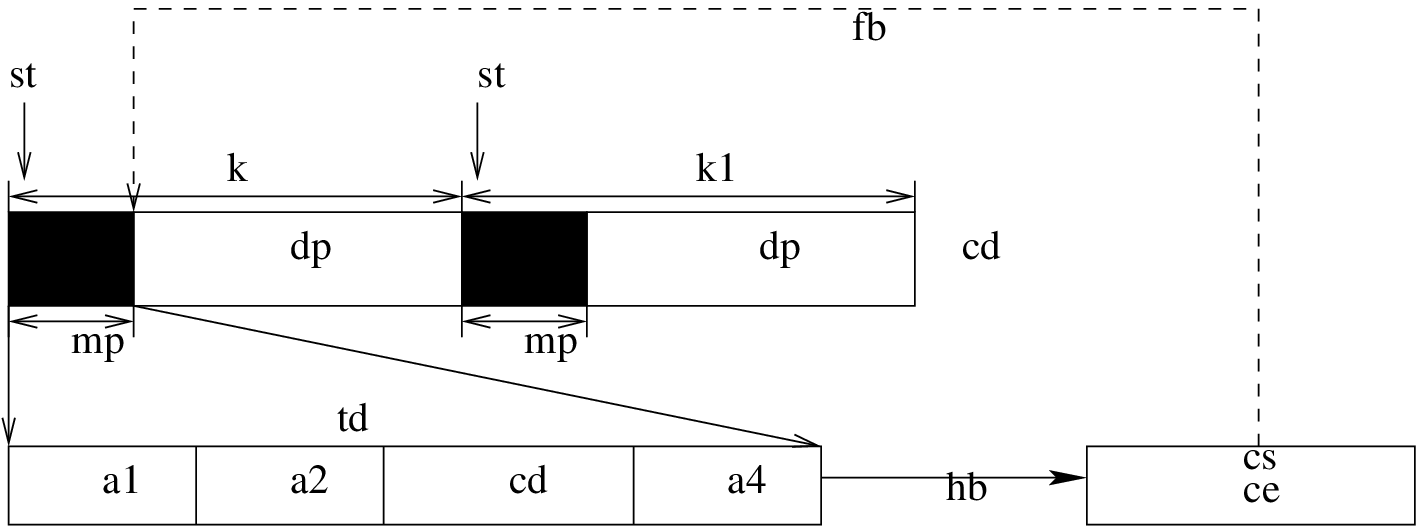} }
    \caption{The overall process of training and data transmission}
    \label{fig:pilots_model}
\end{psfrags}
\end{figure}

Now, the problem is how to optimally design the sequence of pilot beams for
$M_p$ symbol times for each slot.
 Since $M_p \ll N_t$, we can only sense a few columns of  $\widetilde{\Hbf}_{(k)}$ at slot $k$.
  Hence, $M_p$ pilot beams for each slot should be designed
  judiciously by exploiting the channel dynamic and the available
  information in all the previous slots. In the next section, we propose an analytical framework
for optimal pilot beam sequence design  based on  POMDP theory.

\section{POMDP Formulation for Training Beam Sequence Design}
\label{sec:POMDPformulation}

In this section, we formulate a relevant POMDP problem for the
training beam sequence design problem for estimation of large
sparse mmWave MIMO channels.

\subsection{The Action Space and the Observation Space}
\label{sec:actionandobservation}

In the previous section, we assumed that $M_p$ columns of $\Abf_T$
are selected as the pilot beam sequence for the $M_p$ pilot symbol
times for each slot. This is equivalent to selecting $M_p$ columns
of the VCM to be sensed by the training beam sequence at each
slot. Let us denote the selected column indices of the VCM by
\[
\abf = [a_1, a_2, \cdots, a_{M_p}],
\]
 where $a_m$ is the index of the column of the VCM
that is sensed at the $m$-th pilot symbol time. The vector $\abf$ represents the action that we perform at each slot and is referred to as the action vector. Note that there exist ${N_t}\choose{M_p}$ possible $\abf$'s
and the optimal training beam sequence design problem reduces to choosing the
best $\abf$ for each slot under the considered optimality criterion.

After the training period of each slot is finished, the receiver
returns feedback information to the transmitter. The feedback
information contains the detection result about the existence of
paths in the selected columns of the VCM  and  the complex gains
of the detected paths. Then, the transmitter uses the channel gain
information of the detected paths for transmit beamforming during
the data transmission period and uses the feedback information
about the existence of paths to select the training beam indices
for the next slot in an adaptive manner. The second feedback
information can be modeled as
\begin{equation}
    \obf = [o_1, o_2, \cdots, o_{M_p}] \in \{0,1\}^{M_p},
\end{equation}
where $o_m = 1$ indicates that a path is detected by
the training beam transmitted at the $m$-th pilot symbol time, and otherwise $o_m =0$.
Since we have  $2^{M_p}$ possible vectors for $\obf$, the observation
space is given by
\[
{\mathcal{O}}= [\obf^{(1)}, \obf^{(2)}, \cdots, \obf^{(2^{M_p})}].
\]

When the state of the VCM is $\sbf^{(i)}$ and the action vector
$\abf$ is used for the pilot beam sequence for the slot, the
probability that the transmitter observes the feedback information
$\obf^{(j)}$ is denoted as $q_{ij}^\abf$, i.e.,
\begin{align}
     &q_{ij}^{\abf} \triangleq \text{Pr}\{ \obf = \obf^{(j)} | \sbf^{(i)}, \abf\} ~~~\text{for}~ \sbf^{(i)} \in {\mathcal{S}}, \obf^{(j)} \in \mathcal{O}. \label{eq:Pr_feedbackInf}
\end{align}
This probability depends on the detector used to identify the existence of a path at the receiver, and will be discussed next.

\subsubsection{The Channel Sensor at the Receiver}

When the action vector $\abf=[a_1,\cdots,a_{M_p}]$ is used for the current slot,
the receiver filter-bank output at the $m$-th pilot symbol time of the current slot is given from $\eqref{eq:receivedsig}$ by
\begin{equation}  \label{eq:receivedsig11}
\ybf_{(k)}'[m] = \sqrt{P_t}\widetilde{\Hbf}_{(k)}(:,a_m) + \nbf_{(k)}'[m],
\end{equation}
where $\nbf_{(k)}'[m] \sim {\mathcal{CN}}(0,\sigma_N^2\Ibf)$ since
$\Abf_{R}$ is a unitary matrix. Based on $\ybf_{(k)}'[m]$ the
receiver tests the existence of a path at each AoA by checking
each element of the $N_r \times 1$ vector $\ybf_{(k)}'[m]$. Then,
the detection problem for each element of $\ybf_{(k)}'[m]$ is
given by
\begin{equation} \label{eq:pathDetectionProb}
\left\{
\begin{array}{ll}
  {\mathcal{H}}_0: & p(\ybf_{(k)}'[m](n)|\text{empty}) \sim {\mathcal{CN}}(0,\sigma_N^2),\\
  {\mathcal{H}}_1: & p(\ybf_{(k)}'[m](n)|\text{non-empty}) \sim {\mathcal{CN}}(0, P_tN_tN_r\xi^2+\sigma_N^2),
\end{array}
\right.
\end{equation}
where $\ybf_{(k)}'[m](n) =y_R(n) + \iota y_I(n)$ is the $n$-th
element of $\ybf_{(k)}'[m]$, $n=1,2,\cdots,N_r$. We assume that a
Neyman-Pearson detector with size (i.e., false alarm probability)
$P_{FA}$ \cite{Poor:book,Sung&Tong&Poor:06IT} is adopted to test
\eqref{eq:pathDetectionProb}. Then, the detector is given by
\cite{Poor:book}
\begin{equation}
\delta_{NP} = \left\{
\begin{array}{ll}
1, & ~~~|\ybf_{(k)}'[m](n)|^2 ~\ge~ \tau, \\
0, &  ~~~\mbox{otherwise},
\end{array}
\right.
\end{equation}
where $\tau=\sigma_N^2\Gamma^{-1}(1;1-P_{FA})$ since
$|\ybf_{(k)}'[m](n)|^2=y_R^2(n)+y_I^2(n) \sim
\mbox{gamma}(1,1/\sigma_N^2)$ under $\Hc_0$. Here, $\Gamma^{-1}$
is the inverse function of the incomplete gamma function
$\Gamma(x;t)$. The corresponding miss detection probability is
given by \cite{Poor:book}
\begin{equation}  \label{eq:NPPmiss}
P_{MD} = \Gamma \left[ 1; \frac{1}{1+ P_tN_tN_r\xi^2/\sigma_N^2}\Gamma^{-1}(1;1-P_{FA})  \right],
\end{equation}
where $P_tN_tN_r\xi^2/\sigma_N^2$ is the path SNR incorporating the transmit power, path gain, transmit beamforming and receive beamforming.

Now consider $q_{ij}^{\abf}$ in \eqref{eq:Pr_feedbackInf}.
\begin{eqnarray}
q_{ij}^{\abf} &=& \text{Pr}\{ \obf = \obf^{(j)} | \sbf^{(i)}, \abf\}, \nonumber \\
&=& \text{Pr}\{ o_m = o^{(j)}_m, m=1,\cdots,M_p | \sbf^{(i)}, \abf\}, \nonumber\\
&\stackrel{(a)}{=}& \text{Pr}\{ o_1 = o^{(j)}_1 | \sbf^{(i)}, \abf\} \times \cdots \times \text{Pr}\{ o_{M_p} = o^{(j)}_{M_p} | \sbf^{(i)}, \abf\},  \label{eq:qijabfprodindep}\\
&=& \text{Pr}\{ o_1 = o^{(j)}_1 | \sbf^{(i)}, a_1\} \times \cdots \times \text{Pr}\{ o_{M_p} = o^{(j)}_{M_p} | \sbf^{(i)}, a_{M_p}\},
\end{eqnarray}
where step (a) follows since only the noise randomness remains once the state is given, and the receiver noise is assumed to be independent over the $M_p$ pilot symbol times.
$\text{Pr}\{ o_m = o^{(j)}_m | \sbf^{(i)}, a_m\}$ is computed as follows.
\begin{eqnarray}
\text{Pr}\{ o_m = 0 | \sbf^{(i)}, a_m\} &=&  \text{Pr}\{ \delta_{NP}=0 | {\mathcal{H}}_0\}^{N_r - N^{BIN}_{\sbf^{(i)}, a_m}} \text{Pr}\{ \delta_{NP} = 0 | {\mathcal{H}}_1\}^{N^{BIN}_{\sbf^{(i)}, a_m}}  \nonumber \\
&=& (1-P_{FA})^{N_r - N^{BIN}_{\sbf^{(i)}, a_m}} P_{MD}^{N^{BIN}_{\sbf^{(i)}, a_m}},  \label{eq:observationOm0}
\end{eqnarray}
and
\begin{equation}
   \text{Pr}\{ o_m = 1 | \sbf^{(i)}, a_m\} = 1 - \text{Pr}\{ o_m = 0 | \sbf^{(i)}, a_m\},
\end{equation}
where $(P_{FA},P_{MD})$ for the Neyman-Pearson detector is given
in \eqref{eq:NPPmiss}, and $N^{BIN}_{\sbf^{(i)}, a_m}$ is the {\em
actual} number of non-zero bins in the $a_m$-th column of the VCM
when the VCM is in state $\sbf^{(i)}$. Here, two paths with the
same AoD and different AoAs are regarded as two different paths.
\eqref{eq:observationOm0} is obtained because for the AoA
directions with no paths  in the $a_m$-th column of the VCM the
detector should declare $\Hc_0$ and for the AoA directions with
paths in the $a_m$-th column of the VCM the detector should miss
to have $o_m = 0$ (i.e., no non-zero bin in the $a_m$-th column is
declared). The path complex gain can be estimated based on
\eqref{eq:receivedsig11} with a certain estimator such as the
maximum likelihood or minimum mean-square-error (MMSE) estimator.

\begin{figure}[t]
\begin{psfrags}
        \psfrag{block1}[c]{\small $\text{slot}~k$} %
        \psfrag{a1}[c]{\small $\Pbf$}
        \psfrag{a2}[c]{\small $\abf$}
        \psfrag{a3}[c]{\small $\obf$}
        \psfrag{a4}[c]{\small $r_k(\sbf,\abf,\obf)$}
        \psfrag{bleif}[c]{\small $\pibf_k$}
        \psfrag{next}[c]{\small $\pibf_{k+1}$}
        \psfrag{b1}[c]{\small $\text{State}$}
        \psfrag{b2}[c]{\small $\text{transition}$}
        \psfrag{c1}[c]{\small $\text{Pilot beam}$}
        \psfrag{c2}[c]{\small $\text{selection \&}$}
        \psfrag{c3}[c]{\small $\text{transmission}$}
        \psfrag{d1}[c]{\small $\text{Feedback}$}
        \psfrag{d2}[c]{\small $\text{observation}$}
        \psfrag{e1}[c]{\small $\text{}$}
        \psfrag{e2}[c]{\small $\text{Reward}$}
    \centerline{ \scalefig{0.6} \epsfbox{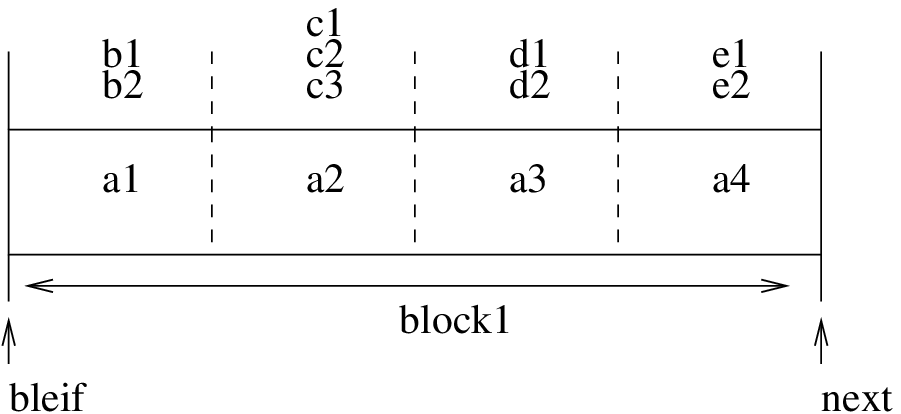} }
    \caption{The sequence of operation at block $k$}
    \label{fig:Thesequenceofoperation}
\end{psfrags}
\end{figure}

\subsection{Sufficient Statistic}
\label{sec:sufficientstatistic}

Fig. \ref{fig:Thesequenceofoperation} depicts the sequence of the
operation at slot $k$ composed of the current belief vector
$\pibf_k$, state transition, action, observation, and following
reward. In the beginning of slot $k$,  the information from all
the past slots is summarized as a belief
vector:\footnote{Following the convention, we define the belief
vector $\pibf_k$ prior to the state transition for each slot, as
shown in Fig. \ref{fig:Thesequenceofoperation}. The belief vector
after the state transition can be updated easily based on the
state transition probability matrix.}
\begin{equation}
    \pibf_k = [\pi_{k,1}, \pi_{k,2}, \cdots, \pi_{k,N}],
\end{equation}
where  $\pi_{k,i}$
is the probability that the state at the beginning of slot $k$ is state $\sbf^{(i)}$ conditioned on all past pilot beam sequences and feedback
information.  It is known
that the belief vector is a sufficient statistic for the action, i.e., the design of the optimal
pilot beam sequence for slot $k$ in our case \cite{Smallwood&Sondic:73OR}.
The transmitter uses the belief vector
to optimally choose the action for slot $k$ from the action space that maximizes the expected reward,
and updates the belief vector for the next block based on the new feedback information at the current slot.

\subsection{The Reward and The Policy}
\label{sec:therewardandpolicy}

Depending on the result of channel identification, a reward is
gained during the data transmission period. According to the
system objective, several
 reward definitions can be considered.
In our case, since the purpose of channel estimation is data
transmission, we consider the number of successfully transmitted
bits over the total considered period of $T$ slots as the final
reward. With the assumption of $\alpha_{(k),l}
\stackrel{i.i.d.}{\sim} \Cc\Nc(0,\xi^2)$ in
\eqref{eq:physical_channelmodel}, the data rate for slot $k$ will
roughly be $\log (1 + N_p N_tN_r\xi^2/\sigma_N^2)$ in case of
maximal ratio combining (MRC) transmission and reception or $N_p
\log (1 + N_tN_r\xi^2/\sigma_N^2)$ in case of spatial multiplexing
used at the transmitter, where $N_p$ is the number of correctly
identified paths. Here, under the assumption of spatial
multiplexing on different paths, the number of successfully
transmitted bits per slot is proportional to $N_p$. Thus, we
choose the number of correctly identified paths as the slot
reward.

If the state of the VCM at slot $k$ is $\sbf^{(i)}$,  the selected
pilot beam sequence at slot $k$ is $\abf$, and the transmitter
observes the feedback information $\obf^{(j)}$, then the immediate
reward\footnote{In case of MRC, we use
$r(\sbf^{(i)},\abf,\obf^{(j)})=\log (1 + [\sum_{m=1}^{M_p}
N^{BIN}_{\sbf^{(i)}, a_m} o_{m}^{(j)}] N_tN_r\xi^2/\sigma_N^2)$
instead.} at slot $k$ can be expressed as
\begin{equation}
    r(\sbf^{(i)},\abf,\obf^{(j)}) = \sum_{m=1}^{M_p} N^{BIN}_{\sbf^{(i)}, a_m}  o_{m}^{(j)}
\end{equation}
where $o_m^{(j)}$ is the $m$-th element of $\obf^{(j)}$. Note that
the false alarm of the receiver detector does not affect the
immediate reward because in this case we have $o_{m}^{(j)}=1$ but
$N^{BIN}_{\sbf^{(i)}, a_m}=0$. (However, there will be no ACK for
the transmitted packet over the slot and the communication
resource will be wasted. The false alarm probability $P_{FA}$ of the
channel sensor should be determined properly by considering this
resource waste.) In the case of miss detection, the opportunity of
successful data transmission is lost. Thus, by reducing the miss
probability $P_{MD}$ of the channel sensor we can obtain more reward.
It is known that the Neyman-Pearson detector minimizes $P_{MD}$ for
given $P_{FA}$ \cite{Poor:book}.

Note that the state at slot $k$ and the feedback information are
unknown at the time of action. Hence, we should consider the
expected reward \cite{PutermanBook}. If the VCM state prior to the
state transition at slot $k$ is $\sbf^{(n)}$ and we perform action
$\abf$, then the immediate expected reward at slot $k$ is given by
\begin{eqnarray}
    R(\sbf^{(n)},\abf) &=&   \sum_{i=1}^N p_{ni}\sum_{j=1}^{2^{M_p}}\text{Pr}\{ \obf = \obf^{(j)} |  \sbf^{(i)}, \abf\}  r(\sbf^{(i)},\abf,\obf^{(j)}) \nonumber \\
            &=&  \sum_{i=1}^N p_{ni}\sum_{j=1}^{2^{M_p}} q_{ij}^\abf \sum_{m=1}^{M_p}N^{BIN}_{\sbf^{(i)}, a_m} o_{m}^{(j)},
\end{eqnarray}
where the state transition from $\sbf^{(n)}$ to all possible $\sbf^{(i)}$ within the slot is captured by $\sum_{i=1}^N p_{ni} (\cdot)$.
 If the belief vector $\pibf_k$ is given at the beginning of slot $k$ and $S_k$ is the random variable representing the state at the beginning of the slot, then the (average) immediate expected reward for action $\abf$ at slot $k$ can be expressed as
\begin{eqnarray}
  \nonumber  {\mathcal{R}}(\pibf_k , \abf ) &=& {\mathbb{E}}\{R(S_k, \abf) | \pibf_k\} \\
                          &=& \sum_{i=1}^N \pi_{k,i}R(\sbf^{(i)},\abf) = \langle\Rbf(\abf),
                          \pibf_k\rangle,
                          \label{eq:immediateExpRew1}
\end{eqnarray}
where  $\Rbf(\abf) := [ R(\sbf^{(1)},\abf),  R(\sbf^{(2)},\abf), \cdots,  R(\sbf^{(N)},\abf)]$,
and $\langle \cdot,\cdot \rangle$ denotes the inner product operation.

In the POMDP framework, a policy $\delta$ is defined as a sequence of functions that
 maps the belief vector to an action for each slot \cite{Smallwood&Sondic:73OR,Monahan:82MS}.  The optimal policy is one that
maximizes the total immediate expected reward accumulated over the
considered total time of $T$ slots\footnote{Such a formulation is
called a finite-horizon POMDP. The considered formulation can be
modified to the infinite-horizon case.} when the initial belief
vector  at the beginning of the transmission is given as
$\pibf_1$. That is, the optimal policy  $\delta^*$ is expressed as
\begin{equation}\label{eq:optimalreward}
    \delta^{*} = \mathop{\arg\max}_{\delta} {\mathbb{E}}_{\delta}\left[\sum_{k=1}^T R(S_k,\abf_k)|\pibf_1\right],
\end{equation}
where $\abf_k$ is the action vector for slot $k$, and ${\mathbb{E}}_\delta$ is the conditional expectation when the policy $\delta$ is given.

\section{Optimal and Suboptimal Strategies for Training Beam Design}
\label{sec:strategy}

\subsection{The Optimal Strategy}
\label{sec:optimalstrategy}

Under the proposed POMDP formulation, the optimal training beam sequence design for maximizing the accumulated data rate over $T$ slots  is equivalent to the problem of finding
the optimal policy satisfying \eqref{eq:optimalreward}.
To solve this problem  we  define the {\it
optimal  value function} $V(\pibf_1)$ as the maximum total expected reward obtained
by the optimal policy $\delta^*$ when the initial belief vector $\pibf_1$ is given at $k=1$:
\begin{equation}
    V(\pibf_1) = {\mathbb{E}}_{\delta^*}\left[ \sum_{k=1}^T R(S_k,\abf_k)|\pibf_1 \right].
\end{equation}
Next consider $V^{k}(\pibf_k)$ defined as the maximum remaining expected reward that can be obtained from
slot $k$ to slot $T$ when a belief vector $\pibf_k$ is given at the beginning of slot $k$. By separating slot $k$ and the remaining slots,  $V^{k}(\pibf_k)$  can be decomposed as  \cite{Ross70applied}
\begin{equation}\label{eq:optimalvaluefunction}
V^k(\pibf_k)   = \max_{\abf_k} \left\{ \langle \Rbf(\abf_k), \pibf_k \rangle  + \sum_{j=1}^{2^{M_p}} V^{k+1}({\mathcal{T}}(\pibf_k | \abf_k,\obf^{(j)})) \gamma(\obf^{(j)}| \pibf_k, \abf_k)  \right\},
\end{equation}
where $\gamma(\obf^{(j)}| \pibf_k, \abf_k) = \sum_{i=1}^N
q_{ij}^{\abf_k} \sum_{n=1}^N \pi_{k,n} p_{ni}$ is the probability
that $\obf^{(j)}$ is observed given belief vector $\pibf_k$ and
action  $\abf_k$ for slot $k$, and ${\mathcal{T}}(\pibf_k |
\abf_k,\obf^{(j)})$ is the updated belief vector from $\pibf_k$ at
slot $k$ for  slot $k+1$ after taking action $\abf_k$ and
observing $\obf^{(j)}$. Here, ${\mathcal{T}}(\pibf_k |
\abf_k,\obf^{(j)})$ can be computed using Bayes's rule as
\cite{Smallwood&Sondic:73OR,Monahan:82MS}
\begin{equation}\label{eq:beliefupdate}
   {\mathcal{T}}(\pibf_k | \abf_k,\obf^{(j)}) =\pibf_{k+1} = [\pi_{k+1, 1}, \pi_{k+1, 2}, \cdots, \pi_{k+1, N}],
\end{equation}
where
\begin{equation}\label{eq:beliefupdate1}
 \pi_{k+1,i} = \text{Pr}\{ \sbf_k = \sbf^{(i)} |\pibf_k, \abf_k, \obf^{(j)}\} = \frac{q_{ij}^{\abf_k} \sum_{n=1}^N \pi_{k,n} p_{ni}}{\sum_{i=1}^N q_{ij}^{\abf_k} \sum_{n=1}^N \pi_{k,n}p_{ni}}, ~~i=1,\cdots,N.
\end{equation}
The first term in the left-hand side (LHS) of \eqref{eq:optimalvaluefunction} is
 the immediate expected reward for slot $k$ and the second term is the the maximum remaining expected reward from slot $k+1$.
As shown in \eqref{eq:optimalvaluefunction}, the selected pilot
beam sequence $\abf_k$ for slot $k$ affects not only the immediate
expected reward at slot $k$ but also the maximum remaining reward
that can be obtained from slot $k+1$. Hence, by considering the
expected reward of the future slots, the performance can be
improved over only  considering the immediate reward at each slot,
i.e., the first  term in the LHS of
\eqref{eq:optimalvaluefunction}.

There exist several known algorithms to obtain the optimal policy
$\delta^*$ over the considered transmission period
$k=[1,2,\cdots,T]$ based on \eqref{eq:optimalvaluefunction} when
the state transition probability $\Pbf$, reward, and observation
and action spaces are given
\cite{Smallwood&Sondic:73OR,Monahan:82MS,Cassandra&Anthony:97UAI}.
For example, point-based POMDP value iteration algorithms are
proposed for  efficiency and low-complexity
\cite{Pineau&Gordon:03IJCAI,Kurniawati&Hsu:08Robtics,Smith&Simmons:12arXiv}.
However, as the number of states and the size of the action space
increase, it requires high computational complexity to obtain the
optimal policy for the POMDP problem even with these point-based
POMDP value iteration algorithms. Thus, to reduce the complexity,
we can alternatively use a suboptimal greedy policy that considers
only the  immediate expected reward at each slot, i.e., the first
term in the LHS of \eqref{eq:optimalvaluefunction}. However, in
large mmWave MIMO channels with large transmit antenna arrays,
even this greedy policy requires high computational complexity due
to the huge numbers  of states and actions. Thus, in the next
subsection we propose a way to implement the greedy policy with
significantly reduced complexity making the proposed POMDP-based
training beam design for large mmWave MIMO channels practical.

\vspace{0.5em}

\begin{remark}
The optimal policy or a suboptimal policy can be computed
off-line once the channel dynamic and other parameters are given,
and the computed policy can be stored beforehand.    Then, in
actual transmission, we start from $k=1$ with $\pibf_1$, and
repeat action and observation until slot $T$. The complexity of
this actual operation is insignificant in general. This is one of
the main advantages of the proposed training beam design approach.
\end{remark}

\subsection{The Proposed Greedy Approach with a Reduced Sufficient Statistic }
\label{sec:reducedbelief}

Although a policy for the POMDP problem can be obtained off-line
as mentioned in the above, it is prohibitive in the case of large mmWave
MIMO channels since the numbers of states and actions grow as
$N_t^L$ and ${N_t}\choose{M_p}$, respectively, in the standard
formulation in the previous sections. For example, when $N_t=64$,
$M_p=10$, and $L=3$, the size of the action space is $1.5\times
10^{11}$ and the number of states is 262,144. The action space
size of $1.5\times 10^{11}$ combined with the state size 262,144
makes solving the POMDP problem prohibitive. Thus, we here focus
on the greedy policy that considers the first term in the LHS of
\eqref{eq:optimalvaluefunction} (i.e., the term in
\eqref{eq:immediateExpRew1}) and its fast solution by introducing
a reduced sufficient statistic that has only $N_t L$ elements.

\begin{proposition}
    With Assumption \ref{ass:lpaths_indep}, let $\omega_{k,l,i}$ be the
    conditional probability (given the pilot sequences and feedback information
    of all past slots) that the $l$-th path is in the $i$-th column of
    the VCM
    at the beginning of slot $k$ for $i=1,\cdots,N_t$. Then,
        \begin{equation}
            \omegabf_k = [\omega_{k,1,1},\omega_{k,1,2}, \cdots, \omega_{k,1,N_t},\omega_{k,2,1},\cdots, \omega_{k,L,N_t}]
        \end{equation}
    is a sufficient statistic to design the optimal pilot sequence
    at slot $k$. Here, $\sum_{li}\omega_{k,l,i}=L$.
\end{proposition}

\begin{proof}
        To prove this statement, we only need to show that $\pibf_k$ can be
expressed by the elements of $\omegabf_k$ since $\pibf_k$ is a
sufficient statistic. Consider $\pi_{k,n}$ for each enumerated
state
    $\sbf^{(n)}$ in $\Sc$.  The state $\sbf^{(n)}$  can be represented by
$(i_1,i_2,\cdots,i_L)$ that means that the $l$-th path is
located in the $i_l$-th column of the VCM, $l=1,\cdots,L$
(see Section
    \ref{sec:multipathscase}). Then, we have
\begin{eqnarray*}
\pi_{k,n} &=& \text{Pr}\{ S_k = (i_1,i_2,\cdots,i_L) |
{\mathcal{I}}_k\}\\
&=& \text{Pr}\{ i_1| {\mathcal{I}}_k\} \times \cdots \times
\text{Pr}\{i_L|
{\mathcal{I}}_k\}\\
&=& \omega_{k,1,i_1}\omega_{k,2,i_2}\cdots \omega_{k,L,i_L},
\end{eqnarray*}
    where $\Ic_k$ represents
    the information from all the previous slots before slot $k$, and the second equality follows from Assumption
    \ref{ass:lpaths_indep} of independent paths.
\end{proof}

\vspace{0.5em} Now we propose a greedy training beam sequence
design algorithm using the reduced belief vector $\omegabf_k$. If
the reduced belief vector $\omegabf_k$ is given at the beginning
of slot $k$ and the pilot beam sequence
 $\abf$ is chosen for the slot, then the immediate expected reward
 $\Rc(\pibf_k,\abf)~(=\langle \Rbf(\abf), \pibf_k \rangle)$
at slot $k$  in \eqref{eq:immediateExpRew1} can be expressed in
terms of $\omegabf_k$ as
\begin{equation}\label{eq:Rewardbyreducedbeliefvector}
    {\mathcal{R}}'(\omegabf_k,\abf) = \sum_{l=1}^{L}\left(\sum_{m=1}^{M_p} \text{Pr}\{ o_m = 1 | \Omega_k(l,a_m) =1\}\sum_{n=1}^{N_t}\omega_{k,l,n}p_{n,a_m}\right)
\end{equation}
due to the independence of the paths, where the binary random
variable $\Omega_k(l,a_m)$ is defined as
\[
\Omega_k(l,a_m) = \left\{
\begin{array}{ll}
1, & \mbox{if the $l$-th path is located at the $a_m$-th column of the VCM at slot $k$},\\
0, & \mbox{otherwise}.
\end{array}
 \right.
\]
Here, $\sum_{n=1}^{N_t}\omega_{k,l,n}p_{n,a_m}$ is the probability
that the $l$-th path is in the $a_m$-th column of the VCM at slot
$k$ after the state transition within the slot, and $\text{Pr}\{
o_m = 1 | \Omega_k(l,a_m) =1\}$ is the probability that the
receiver detects the $l$-th path successfully when the $l$-th path
is in the $a_m$-th column of the VCM. $\text{Pr}\{ o_m = 1 |
\Omega_k(l,a_m) =1\}$ in \eqref{eq:Rewardbyreducedbeliefvector} is
related to the miss detection probability of the receiver
detector.  To obtain a tractable expression for the immediate
expected reward, we assume that we have a reasonable path sensing
SNR $N_tN_rP_t \xi^2/\sigma_N^2$ and the miss detection
probability $P_{MD}$ of the receiver detector is zero, i.e.,
$\text{Pr}\{ o_m = 1 | \Omega_k(l,a_m) =1\} =1$.
 Then,
\eqref{eq:Rewardbyreducedbeliefvector} can be rewritten as
\begin{eqnarray}
 {\mathcal{R}}'(\omegabf_k,\abf) &=& \sum_{l=1}^{L}\sum_{m=1}^{M_p} \sum_{n=1}^{N_t}\omega_{k,l,n}p_{n,a_m}   \nonumber  \\
                                 &=& \sum_{l=1}^{L}\sum_{n=1}^{N_t} \omega_{k,l,n}\sum_{m=1}^{M_p} p_{n,a_m}   \nonumber \\
                                 &=& \left[ \underbrace{\sum_{m=1}^{M_p} p_{1,a_m} ,\cdots,\sum_{m=1}^{M_p} p_{N_t,a_m}}_{\mbox{1st path}} ,\cdots,
                                                 \underbrace{\sum_{m=1}^{M_p} p_{1,a_m}, \cdots,\sum_{m=1}^{M_p} p_{N_t,a_m}}_{\mbox{$L$-th path}}   \right] \omegabf_k   \nonumber \\
                                  &\stackrel{(a)}{=}& \sum_{m \in
                                  \abf}\Pbf'(m,:)\omegabf_k
                                  \label{eq:reducedSufStatRewardNew}
\end{eqnarray}
where $\Pbf' := [\underbrace{\Pbf^T}_{\mbox{1st}}, \Pbf^T, \cdots, \underbrace{\Pbf^T}_{\mbox{$L$-th}}]$ with $\Pbf$ defined in \eqref{eq:transitionprobability}, and $\Pbf'(m,:)$ denotes the $m$-th row of $\Pbf'$. (In step (a), the summation went out of the bracket.)
Therefore, the greedy training beam sequence design problem of choosing $\abf$ that maximize the immediate expected reward at slot $k$ reduces to choosing $\abf$ such that
\begin{equation}\label{eq:reducedoptimal}
    \abf^* = \mathop{\arg\max}_{\abf}  \sum_{m \in \abf}\Pbf'(m,:)\omegabf_k
\end{equation}
Thus, for the optimal choice of the training beam sequence for the greedy policy, from \eqref{eq:reducedoptimal}, we only need to find the set of the indices that correspond to the
largest $M_p$ values in the $N_t \times 1$ vector $\Pbf'\omegabf_k$, which can be obtained by $N_t^2L$ real multiplications.
Note that in the proposed solution we do not need to consider ${N_t}\choose{M_p}$ complexity for $\abf$, but one time sorting of an $N_t \times 1$ vector is enough!

At the end of slot $k$, the transmitter updates the reduced belief
vector from $\omegabf_k$ to $\omegabf_{k+1}$ based on the selected
pilot beam sequence and the observed feedback information at the
current slot. This update process is given in the following
proposition.

\begin{proposition} With the action vector $\abf$ and observation $\obf^{(j)}$ for slot $k$, the updated reduced sufficient statistic is given by
\begin{equation}\label{eq:reducedupdate}
   {\mathcal{T}}(\omegabf_k | \abf, \obf^{(j)}) := \omegabf_{k+1}  = [\omega_{k+1,1,1}, \omega_{k+1,1,2}, \cdots, \omega_{k+1,L,N_t}],
\end{equation}
where each element of $\omegabf_{k+1} $ is obtained by Bayes's rule as
\begin{equation}\label{eq:reducedupdate1}
 \omega_{k+1,l,i} =\frac{{q'}_{ijl}^\abf \sum_{n=1}^{N_t} \omega_{k,l,n}p_{ni}}{\sum_{i=1}^{N_t}{q'}_{ijl}^\abf \sum_{n=1}^{N_t} \omega_{k,l,n}p_{ni}},
\end{equation}
and  ${q'}_{ijl}^\abf$ is
the probability that the $j$-th feedback information $\obf^{(j)}$ is observed
conditioned on that the $l$-th path is located in the $i$-th column of the VCM after state transition within the slot
 and
the pilot beam sequence $\abf$ is transmitted.
The quantity ${q'}_{ijl}^\abf$ is given by
\begin{equation}\label{eq:q'define}
    {q'}_{ijl}^{\abf} = \prod_{m=1}^{M_p} \text{Pr} \{o_m = o_m^{(j)} | \Omega_{k+1}(l,i)=1, \abf, \omegabf_k\}.
\end{equation}
where $\text{Pr}\{o_m^{(j)} = 1 |
\Omega_{k+1}(l,i)=1,\abf,\omegabf_k \}=1-\text{Pr}\{o_m^{(j)} = 0|
\Omega_{k+1}(l,i)=1,\abf,\omegabf_k \}$, and
\begin{align}\label{eq:Pr_observation}
  & \text{Pr}\{o_m^{(j)} = 0 | \Omega_{k+1}(l,i)=1,\abf,\omegabf_k \}   \nonumber \\
             & ~~~~~~~~~~~=\left\{
                    \begin{array}{ll}
                        P_{MD}^\prime, & \text{if}~ a_m = i, \\
                        \prod_{s =1 , s\neq l }^L \left( \sum_{t=1}^{N_t} \omega_{k,s,t}(1- p_{t a_m})\right)(1- P_{FA}) & \\
                        \quad + \left(1 -\prod_{s =1 , s\neq l }^L \left( \sum_{t=1}^{N_t} \omega_{k,s,t}(1- p_{t a_m})\right)\right) P_{MD}^\prime,  & \text{if}~ a_m \neq i.
                    \end{array}
                \right.
\end{align}
Here, $P_{MD}^\prime$ is the probability of miss detection at the
channel sensor when {\em at least} one path exists in the sensed
column of the VCM.
\end{proposition}

\vspace{1em}
\begin{proof}
    Applying Bayes's formula, we express $\omega_{k+1,l,i}$ in \eqref{eq:reducedupdate} as     \begin{eqnarray}\label{eq:reducedbayes}
        \nonumber \omega_{k+1,l,i} &=& \text{Pr}\{\Omega_{k+1}(l,i) =1 | \obf^{(j)}, \abf, \omegabf_k \}\\
                                   &=& \frac{\text{Pr}\{\Omega_{k+1}(l,i) =1,  \obf^{(j)} | \abf, \omegabf_k \}}{\sum_{i=1}^{N_t}\text{Pr}\{\Omega_{k+1}(l,i) =1,  \obf^{(j)} | \abf, \omegabf_k \} }.
    \end{eqnarray}
    The numerator of \eqref{eq:reducedbayes} is expressed as follows:
    \begin{eqnarray}
    \nonumber     && \text{Pr}\{\Omega_{k+1}(l,i) =1, \obf = \obf^{(j)} | \abf, \omegabf_k \} \\
    \nonumber     &&~~~ \quad =\sum_{n=1}^{N_t} \text{Pr}\{ \Omega_k(l,n) = 1 , \Omega_{k+1}(l,i) =1, \obf = \obf^{(j)} | \abf, \omegabf_k \} \\
    \nonumber     &&~~~ \quad =\sum_{n=1}^{N_t} \text{Pr}\{ \Omega_{k}(l,n)=1 |\abf,\omegabf_k \}  \cdot \text{Pr}\{ \Omega_{k+1}(l,i)=1,\obf= \obf^{(j)} | \Omega_k(l,n) =1 ,\abf,\omegabf_k \} \\
    \nonumber    &&~~~ \quad =\sum_{n=1}^{N_t} \text{Pr}\{ \Omega_{k}(l,n)=1 |\abf,\omegabf_k \}  \cdot \text{Pr}\{ \Omega_{k+1}(l,i)=1| \Omega_k(l,n) =1 ,\abf,\omegabf_k \} \\
    \nonumber    &&~~~ \quad \quad \quad ~~~ \cdot \text{Pr}\{ \obf = \obf^{(j)} | \Omega_{k+1}(l,i)=1, \Omega_k(l,n) =1 ,\abf,\omegabf_k \}\\
                 &&~~~ \quad = \sum_{n=1}^{N_t} \omega_{k,l,n} \cdot p_{ni} \cdot  \text{Pr}\{ \obf = \obf^{(j)} | \Omega_{k+1}(l,i)=1, \abf,\omegabf_k \},
    \end{eqnarray}
    where the last step follows by the definitions of $\omega_{k,l,n}$ and
    $p_{ni}$ and the Markovian assumption on the state transition.
    Denoting  $\text{Pr}\{ \obf = \obf^{(j)} | \Omega_{k+1}(l,i)=1, \abf,\omegabf_k \}$ by ${q'}_{ijl}^{\abf}$, we have \eqref{eq:reducedupdate1}.

    As in deriving  \eqref{eq:qijabfprodindep}, once the state is given, the remaining randomness is the receiver noise. Thus, by the assumption of independent receiver noise over symbol time, we have
    \begin{eqnarray}
    \nonumber     {q'}_{ijl}^{\abf} &:=&  \text{Pr}\{ \obf = \obf^{(j)} | \Omega_{k+1}(l,i)=1, \abf,\omegabf_k \} \\
                          &=& \prod_{m=1}^{M_p} \text{Pr}\{ o_m = o_m^{(j)} | \Omega_{k+1}(l,i) =1, \abf, \omegabf_k\}.
    \end{eqnarray}
If the $l$-th path is in the $i$-th column of the VCM  and the
$i$-th column of the VCM is chosen to be sensed by the $m$-th
pilot beam in the slot (i.e., $a_m=i$),
    the probability that $o_m^{(j)}=0$ is observed  is
    equal to the miss detection probability $P_{MD}^\prime$ of the channel sensor. This explains the first part of \eqref{eq:Pr_observation}.
    If the $m$-th pilot beam
    senses
    a column that does not contain the $l$-th path located at the $i$-th column of the VCM (i.e., $a_m \ne i$), then the probability
    that $o_m^{(j)}=0$ is observed
    is equal to
    \begin{equation}  \label{eq:lemma1reducedLast}
           \prod_{s =1 , s\neq l }^L \left( \sum_{t=1}^{N_t} \omega_{k,s,t}(1- p_{t a_m})\right)(1- P_{FA})
        + \left(1 -\prod_{s =1 , s\neq l }^L \left( \sum_{t=1}^{N_t} \omega_{k,s,t}(1- p_{t a_m})\right)\right)
        P_{MD}^\prime.
   \end{equation}
    The first term in \eqref{eq:lemma1reducedLast}  is the probability
    that no other path than the $l$-th path is located in the $a_m$-th column of
    the VCM
    and a false alarm does not occur, and the second term  in \eqref{eq:lemma1reducedLast} is
    the probability that at least one path other than the $l$-th path exists in the $a_m$-th column of the VCM
    but the channel sensor misses this path. This explains the second part of \eqref{eq:Pr_observation}.
\end{proof}

For simplicity, we use $P_{MD}$ in the single path case instead of
$P_{MD}^\prime$ neglecting  the case that  there  can exist more
than one path in the sensed column of the VCM.   However, this
will have a negligible impact on the performance because the event
that more than one path are located in the same column of the VCM
will rarely occur due to the sparsity of large mm-Wave MIMO
channels. Note that the action vector computation and the belief
vector update required for the proposed greedy algorithm based on
the reduced sufficient statistic can be run with insignificant
computational complexity.

\subsection{Extensions}
\label{sec:extensions}

{\em Multi-resolution:}  Until now we have assumed that the
transmitter picks one column of $\Abf_T$ as the pilot beam at each
pilot symbol time and search one column of the VCM
 at each pilot symbol time,  i.e., $\xbf_n
\in \{\sqrt{P_t}\abf_{TX}(\tilde{\theta}^t_{1})$, $\cdots$,
$\sqrt{P_t}\abf_{TX}(\tilde{\theta}^t_{N_t}) \}$. This assumption
can be relaxed depending on the path SNR. Recall that the path SNR
is given by $N_tN_rP_t \xi^2/\sigma_N^2$, where $\xi^2$ is the
path gain square and $\sigma_N^2$ is the noise variance. When the path SNR is
 sufficiently high, we can use a superposed pilot beam at each symbol time, i.e.,
\[
\xbf_n \in \left\{ \sum_{i=1}^{N_a}\sqrt{\frac{P_t}{N_a}}\abf_{TX}(\tilde{\theta}^t_{i}), \sum_{i=N_a+1}^{2N_a}\sqrt{\frac{P_t}{N_a}}\abf_{TX}(\tilde{\theta}^t_{i}),\cdots, \sum_{i=N_t-N_a+1}^{N_t}\sqrt{\frac{P_t}{N_a}}\abf_{TX}(\tilde{\theta}^t_{i})\right\},
\]
where $N_a$ consecutive columns of $\Abf_T$ are superposed as one pilot beam for a symbol time under the assumption that $N_a$ is a divisor of $N_t$. This is equivalent to effectively  making the search grid in the VCM small and shortening the initial lock-in time. The corresponding state transition probability matrix can be obtained from the original state transition matrix $\Pbf$. Once the non-zero paths in the smaller grid are detected, we can switch to the original high-resolution narrow pilot beam by properly assigning the belief vector. Here, $N_a$ should be determined appropriately so that
$N_tN_rP_t \xi^2/(N_a \sigma_N^2)$ is high enough to make the probability of miss detection by each superposed pilot beam  small.

{\em 2-dimensional search with limited receiver beamforming:} We
have assumed that the receiver has a filter bank that searches all
AoA directions. This requires $N_r$ RF chains at the receiver.
When the receiver has one RF chain, typically the beamforming is
implemented at the antenna array. That is, the antenna array has
multiple receive beamforming directions and switching is performed
among the receive beamforming directions. In this case, the
proposed POMDP framework can be extended to search the non-zero
paths in a truly 2-dimensional grid.

{\em The mult-user case:} Until now we have considered the optimal
training beam design for channel estimation and tracking via POMDP
in single-user MIMO systems. We can apply our POMDP formulation to
downlink multi-user MIMO channels with simple modifications. The
MIMO channel with multiple users can be considered as the MIMO
channel with more paths. In the multi-user case,  the transmitter
needs to know which user each path belongs to. This identification
can be delivered in feedback information from users, and the POMDP
formulation with this additional feedback information can be used
in the multi-user MIMO case.

\section{Numerical results}
\label{sec:NumericalResult}

In this section, we provide some numerical results to evaluate the
performance of the proposed pilot design and channel estimation
method. Throughout the simulation, we fixed the number of paths as
$L=2$, the path gain $\xi^2=1$, and the noise power
$\sigma_N^2=1$;   set the transmit power $P_t$ in each case
so that the path SNR $N_tN_rP_t \xi^2/\sigma_N^2$ is 20 dB; and used a Neyman-Pearson detector for the channel sensor at the receiver.   We
assumed that each path behaves according to the state transition
probability matrix \eqref{eq:bandedstructure_random} under
Assumption \ref{ass:lpaths_indep}.

\begin{figure}[!t]
\centerline{ \SetLabels
\L(0.5*-0.1) (a) \\
\endSetLabels
\leavevmode
\strut\AffixLabels{
\scalefig{0.6}\epsfbox{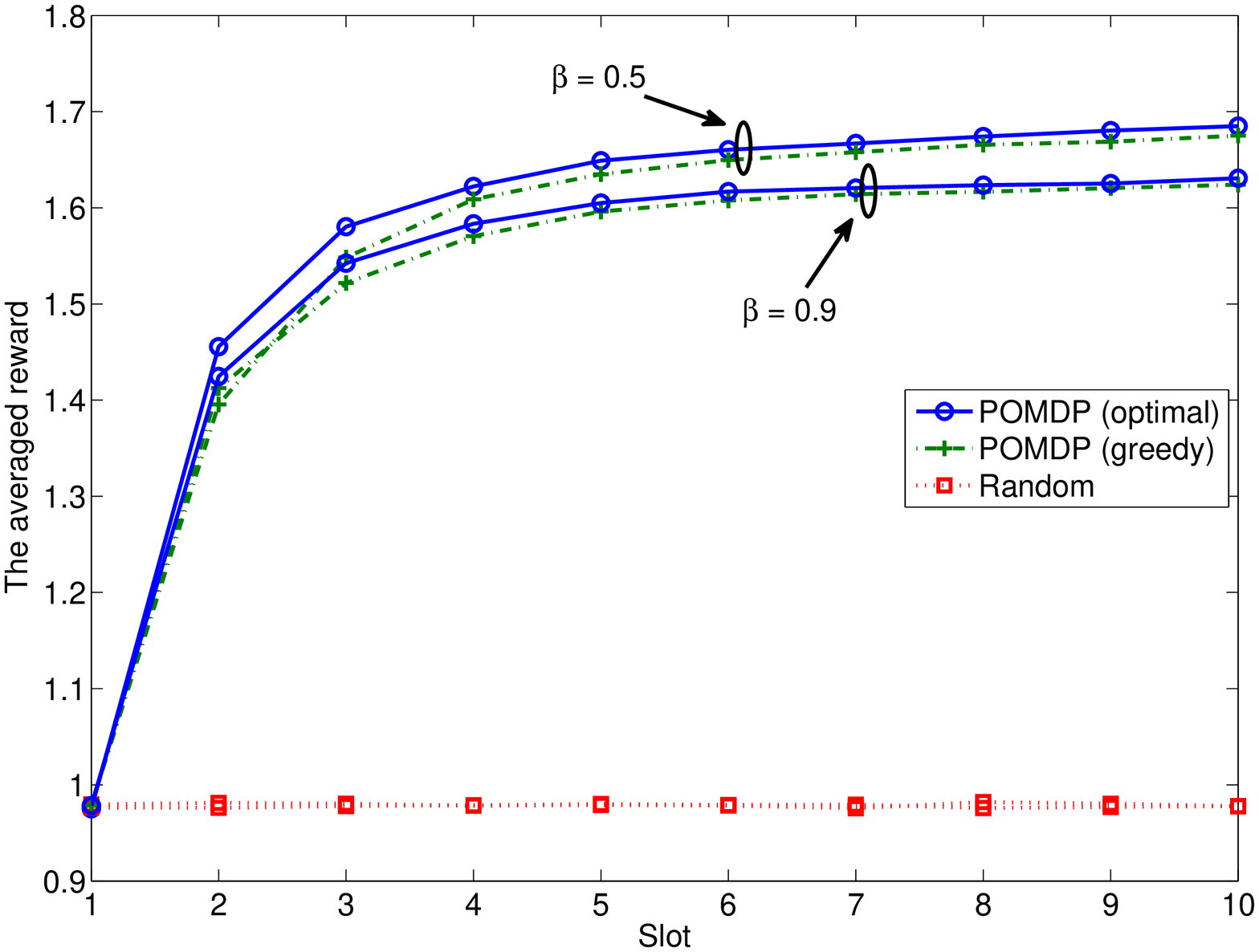}
 } } \vspace{0.8cm}
\centerline{ \SetLabels
\L(0.5*-0.1) (b) \\
\endSetLabels
\leavevmode
\strut\AffixLabels{
\scalefig{0.6}\epsfbox{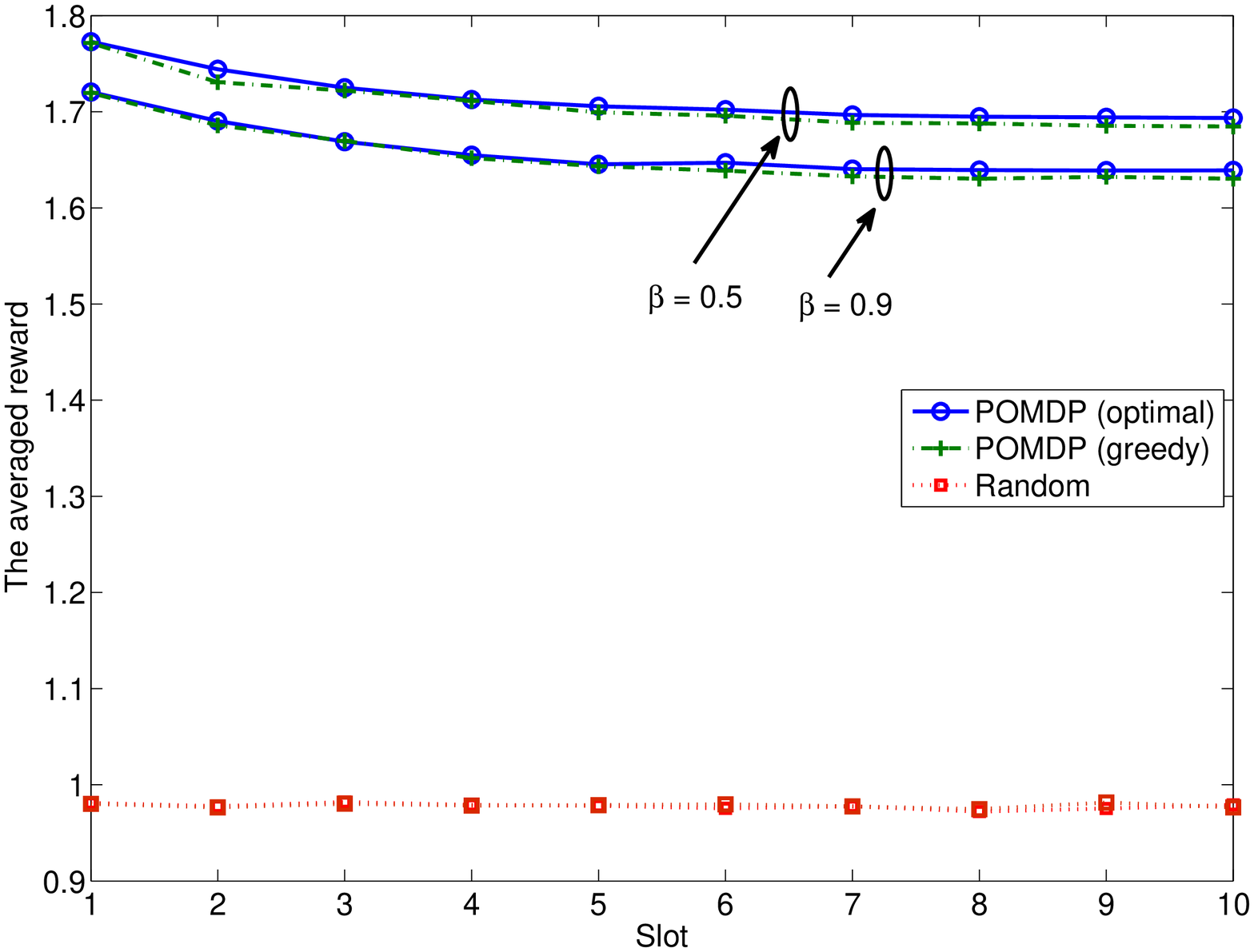}
 } } \vspace{0.75cm}
\caption{Averaged per-slot reward with respect to time ($N_t =8,
~N_r = 4, ~M_p = 4, ~L =2, ~\lambda =0, ~M_s=10$)}
\label{fig:Averagedreward}
\end{figure}

First, we considered a small MIMO system with $N_t = 8$, $N_r =
4$, and $M_p =4$ so that the optimal policy can be computed with
reasonable complexity. Here, we set the false alarm probability of the channel sensor for each AoA direction as $P_{FA}=0.05$, and  $B=1$ and $\lambda=0$ for
$\Pbf_{\beta,\lambda}^B$ in \eqref{eq:bandedstructure_random}, and
considered two different values of $\beta$: $\beta=0.5$ (slow
fading) and $\beta=0.9$ (fast fading). Fig.
\ref{fig:Averagedreward} shows the performance of three methods in
this case: the optimal POMDP strategy, the greedy POMDP strategy
and a random selection strategy that selects $M_p$ pilot beams
randomly at each slot. The $y$-axis shows
 the averaged per-slot reward, i.e., the
averaged number of detected paths for each slot. (The per-slot
reward shown in the figure is the average over 100,000
realizations of the specified Markov random processes of length 10
slots.) We first tested the initial lock-in behavior of the
three methods. Here, assuming no {\it a priori} state
information in the beginning, we set all elements of the initial belief vector  to
be equal.  Fig. \ref{fig:Averagedreward} (a) shows the initial
lock-in behavior in this case. On the other hand, Fig.
\ref{fig:Averagedreward} (b) shows the channel tracking behavior
of the three methods. Here, we assumed the exact state
knowledge at the first slot.  It is seen  that the optimal POMDP
method and the greedy POMDP method significantly outperform the
random selection strategy and   the performance difference between
the optimal strategy and the greedy strategy is not significant.
It is also seen that as expected, the channel estimation and
tracking performance is better at $\beta=0.5$ than at $\beta=0.9$.

\begin{figure}[!t]
\centering
\scalefig{0.65}\epsfbox{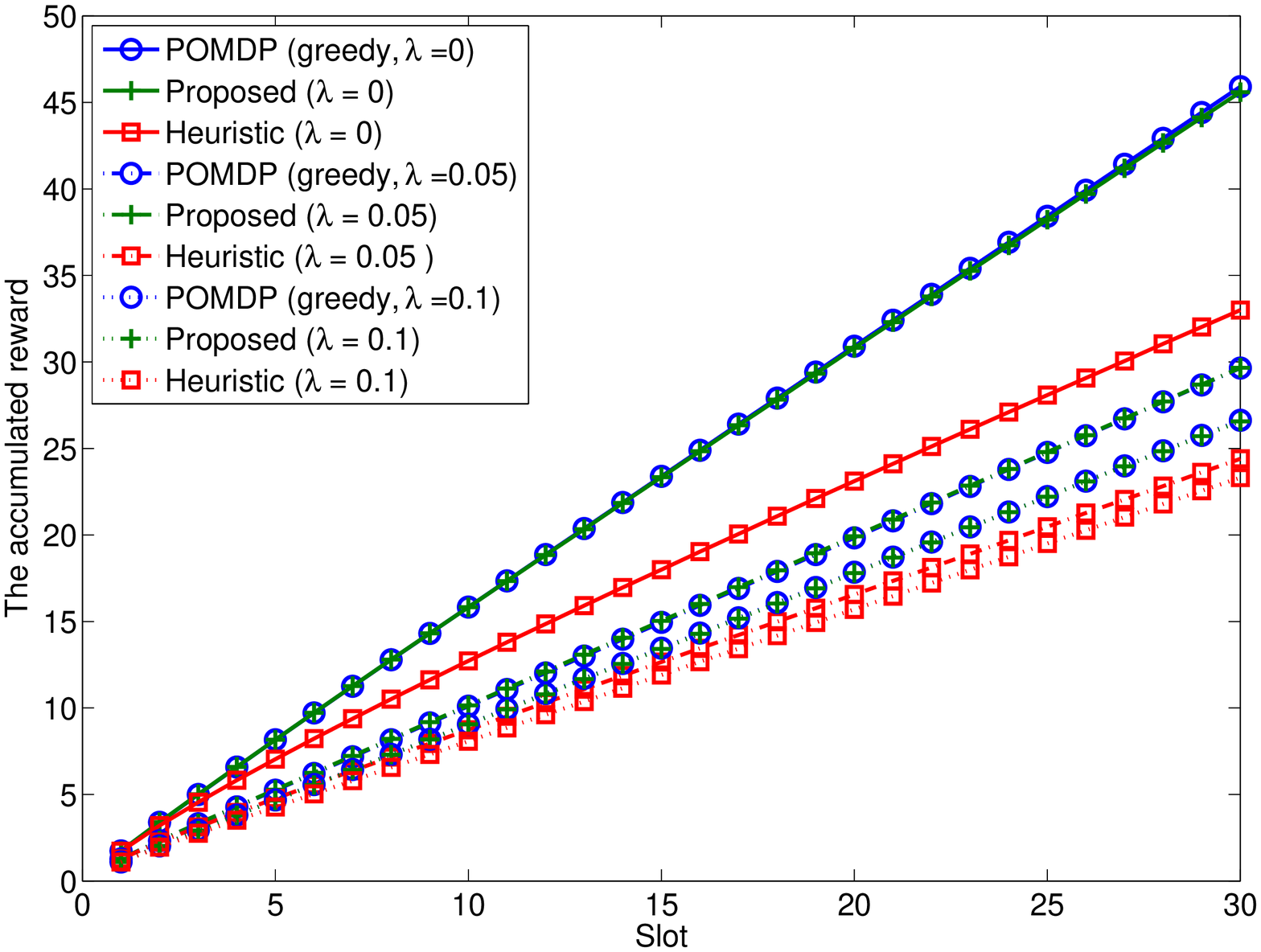}
\caption{Accumulated reward with resect to time ($N_t =16$, $N_r =
4$, $M_p = 6$, $L =2$, $\beta = 0.5$, $M_s=30$)}
\label{fig:accumulated_reward_middleantennaarray}
\end{figure}

\begin{figure}[!h]
\centering
\scalefig{0.65}\epsfbox{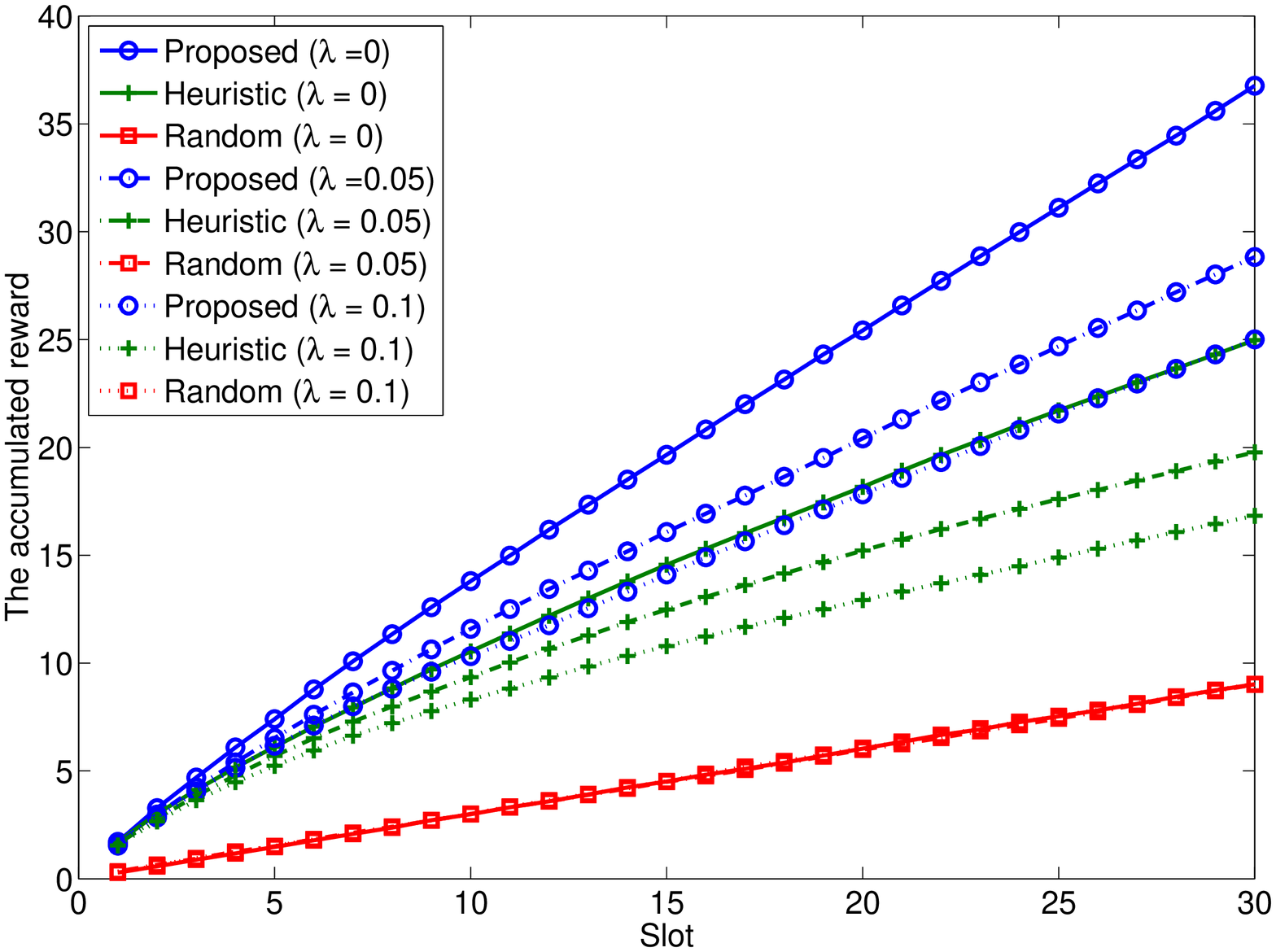}
\caption{Accumulated reward with resect to time ($N_t =64$, $N_r =
16$, $M_p = 10$, $L =2$, $\beta = 0.5$, $M_s=30$)}
\label{fig:accumulated_reward_largeantennaarray}
\end{figure}

Second, we considered a  larger MIMO system with $N_t = 16$, $N_r
= 4$, $M_p = 6$, and $M_s=30$ for which the complexity of
computing the optimal policy is already high. Here, we considered
three pilot beam design and channel estimation methods: the exact
greedy POMDP method, the greedy POMDP method with the proposed
reduced-size sufficient statistic, and a heuristic method. For the
state transition probability matrix, we used
$\Pbf_{\lambda,\beta}^B$ in \eqref{eq:bandedstructure_random} with
$B=2$, $\beta=0.5$ and  varied $\lambda$. We set the false alarm probability of the channel sensor as $P_{FA}=0.05$.
 For the heuristic
method, $3$ pilot symbol times out of $M_p=6$ pilot symbol times
per slot are used to track one path and the remaining $3$ pilot
symbol times are used to track the other path based on the
detection result of the previous slot.  The heuristic method
operates as follows. If a path is detected successively at a slot,
the pilot beams of the next slot are adaptively determined to
track the path by considering the probability of the path's
movement based on the state transition probability matrix. If not,
the pilot beams are unchanged for the next slot.  Fig.
\ref{fig:accumulated_reward_middleantennaarray} shows the averaged
accumulated reward of the three methods with respect to
 time under the assumption that the initial state is known.
The performance in the figure is the average over 10,000
realizations of the specified Markov process of length 30 slots.
Recall that in the proposed reduced-complexity greedy method, the
exact immediate reward \eqref{eq:Rewardbyreducedbeliefvector} is
approximated by  \eqref{eq:reducedSufStatRewardNew}.  It is seen
that the impact of this approximation is negligible and the
proposed reduced-complexity greedy POMDP method performs as well
as the exact greedy POMDP method. It is also seen that the two
POMDP-based greedy methods outperforms the heuristic method that
still incorporates the state transition in a simple manner.
\begin{figure}[!t]
\centerline{ \SetLabels
\L(0.5*-0.1) (a) \\
\endSetLabels
\leavevmode
\strut\AffixLabels{
\scalefig{1}\epsfbox{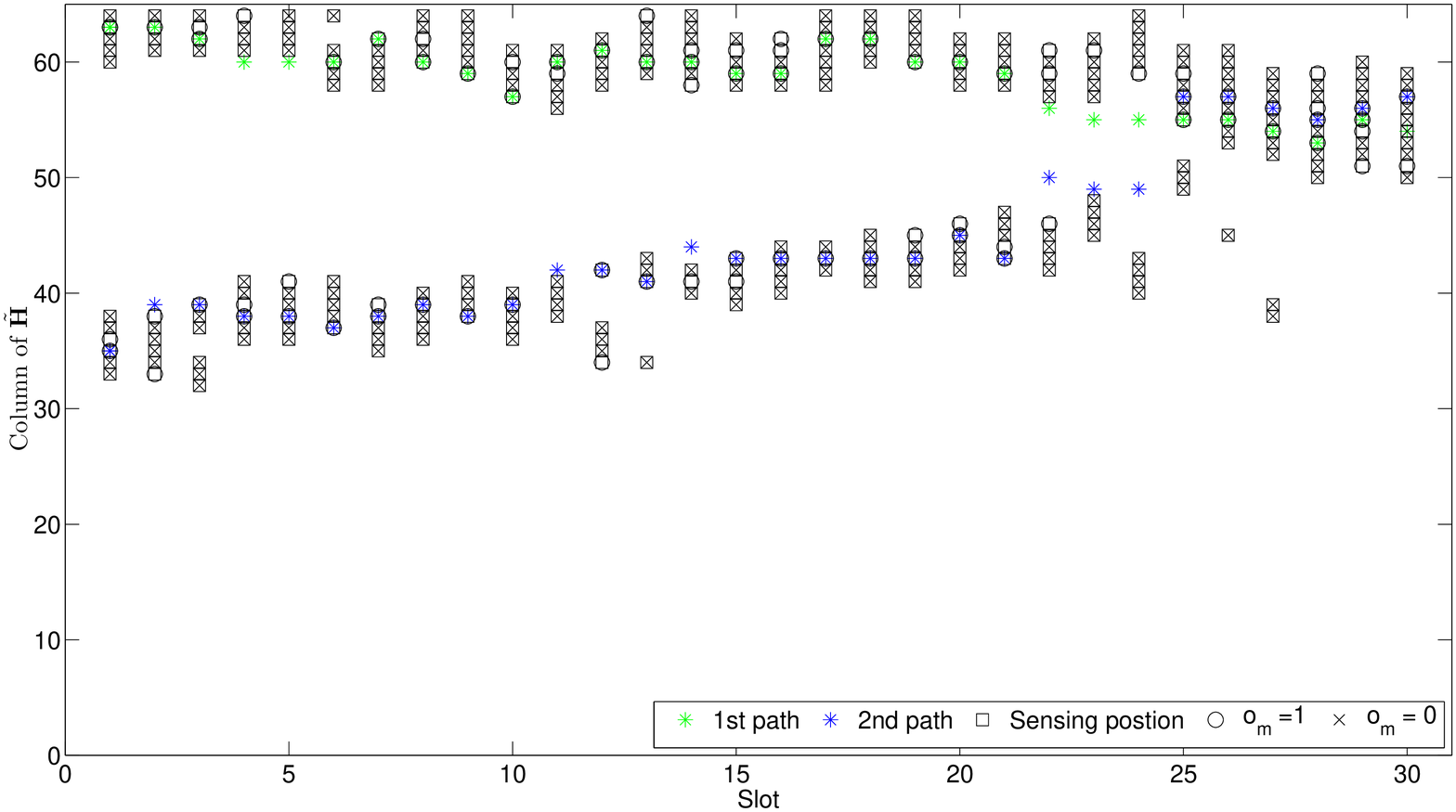}
 } } \vspace{0.8cm}
\centerline{ \SetLabels
\L(0.5*-0.1) (b) \\
\endSetLabels
\leavevmode
\strut\AffixLabels{
\scalefig{1}\epsfbox{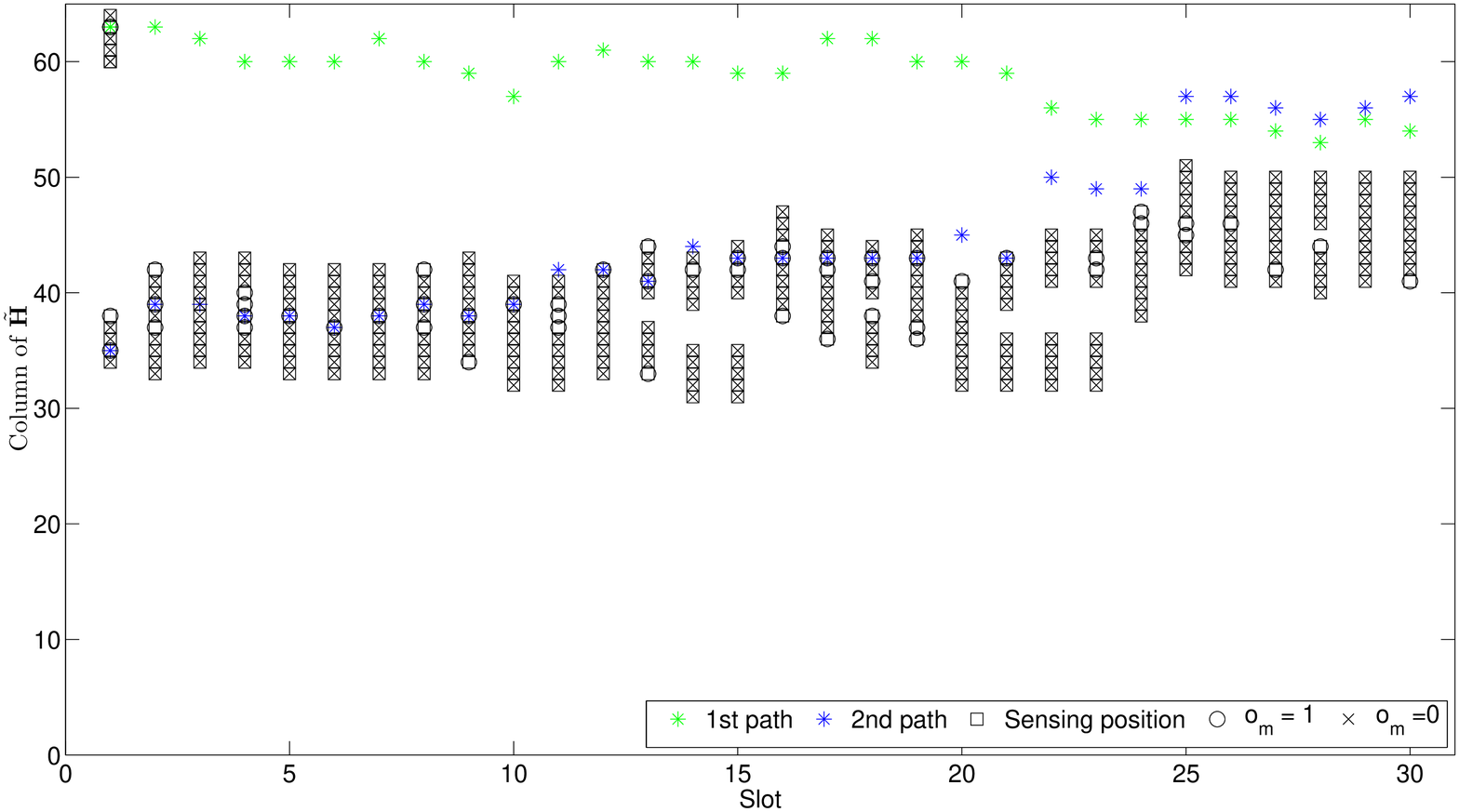}
 } } \vspace{0.75cm}
\caption{The tracking process in the MIMO channel ($N_t =64, N_r = 16, M_p = 10, L =2, \lambda =0$)
 (a): The proposed strategy and (b) : The heuristic strategy.}
\label{fig:Single realization}
\end{figure}

Next, we considered a large MIMO channel with
 $N_t = 64$, $N_r=16$, and $M_p = 10$. Here, we set $B=8$,
 $\beta=0.5$ and varied $\lambda$ for $\Pbf_{\beta,\lambda}^B$.
In
 this case, we only considered the proposed reduced-complexity greedy POMDP
 method and the simple heuristic method that have feasible
 complexity.  Due to the large $N_r$, we set the false alarm probability of the channel sensor as $P_{FA}=0.01$ in this case. (If the per-AoA-direction false alarm probability is not sufficiently small in case of large $N_r$, the overall false alarm probability combining the sensing results of all the AoA directions is too high.)
  Fig. \ref{fig:accumulated_reward_largeantennaarray} shows
the accumulated reward performance (with a known initial state)
with respect to time (averaged over 500 realizations) in this
case. Similar behavior is observed as in Fig.
\ref{fig:accumulated_reward_middleantennaarray}.  Fig.
\ref{fig:Single realization} shows the tracking performance in a
single realization of the Markov channel process of length 30
slots. In Fig. \ref{fig:Single realization}, the green and blue
stars indicate the positions of the two paths at each slot, and
$\Box$'s indicate
 the pilot beam indices at each slot, and $\circ$'s
and $\times$'s indicate the feedback information of path presence
and path absence, respectively, in the sensed index.
 It is seen that the proposed method tracks
the two paths successively, whereas the heuristic strategy is
ineffective. In the case of the heuristic method, it is difficult
to track the path if the path is missed in the previous block, and
it is not possible to distinguish a false alarm.  However, the
proposed POMDP method tracks the paths in each slot by updating
the reduced-size belief vector based on the history, and the false
alarm can be handled in the updating process of the reduced-size
belief vector. Thus, the proposed POMDP-based method is much more
effective in sparse large MIMO channel tracking.

\section{Conclusion}
\label{sec:conclusion}

We have considered the adaptive training beam sequence design
problem for large sparse  mmWave MIMO channels. By imposing a
Markov random walk assumption on the paths' movement, we have
formulated the adaptive training beam sequence design problem as a
POMDP problem. Under this POMDP framework, we have obtained
optimal and suboptimal strategies for adaptive training beam
sequence design  for large sparse  mmWave MIMO channels.
Furthermore, we have derived a fast greedy POMDP algorithm for
this problem with significantly reduced complexity for practical
implementation by proposing a new sufficient statistic for the
decision process. We have considered several extensions of the
proposed method. The proposed training beam design and channel
estimation method can be applied both to the initial channel
learning and to efficient channel tracking after the
channel is initially learned by some other method.

\begin{thebibliography}{}


\bibitem{Venkateswaran&Veen:10SP} V. Venkateswaran and A. van der Veen ``Analog beamforming in MIMO communications with phase shift networks and online channel estimation,"
{\it IEEE Trans. Signal Process.}, vol. 58, no. 8, pp. 4131 -- 4143, 2010.

\bibitem{Ayach&HeathEtAl:12ICC} O. E. Ayach, R. W. Heath Jr., S. Abu-Surra,  S. Rajagopal, and Z. Pi, ``Low complexity precoding for large millimeter wave MIMO systems," in
{\it Proc. IEEE Int. Conf. Commun.}, Jun. 2012.

\bibitem{Alkhateeb&Ayach&Leus&Heath:13ITA} A. Alkhateeb, O. E. Ayach, G. Leus, and R. W. Heath, ``Hybrid precoding for millimeter wave cellular systems with partial channel knowledge," in
{\it Proc. Inf. Theory and Appl. Workshop.}, (San Diego, CA), 2013.

\bibitem{Alkhateeb&Ayach&Leus&Heath:14arxiv} A. Alkhateeb, O. E. Ayach, G. Leus, and R. W. Heath, ``Channel Estimation and Hybrid Precoding for Millimeter Wave Cellular Systems,"
{\it submitted to IEEE J.Sel. Topics Signal Process., arXiv preprint arXiv:1401.7426,2014}.

\bibitem{Doan&Emami&Sobel:04ComMag} C. Doan, S. Emami, D. Sobel, A. Noknejad and R. Brodersen, ``Design considerations for 60 GHz CMOS radios,"
{\it IEEE Commun. Mag.}, vol. 42, no. 12, pp. 132 -- 140, 2004.

\bibitem{Telatar:00Euro} I. E. Telatar, ``Capacity of multi-antenna Gaussian channels,"
{\it European Trans. Telecommun.}, vol. 10, no. 6, pp. 585 -- 595, 1999.

\bibitem{Hassibi&Hochwald:03IT} B. Hassibi and B. M. Hochwald, ``How much training is needed in multiple-antenna wireless links?,"
{\it IEEE Trans. Inf. Theory}, vol. 49, pp. 951-- 963, Apr. 2003.

\bibitem{Kotecha&Sayeed:04SP} J. H. Kotecha and A. M. Sayeed, ``Transmit signal design for optimal estimation of correlated MIMO channels,"
{\it IEEE Trans. Signal Process.}, vol. 52, pp. 546 -- 557, Feb. 2004.

\bibitem{Choi&Love&Bidigare:14JSTSP} J. Choi, D. J. Love, and P. Bidigare, ``Downlink training techniques for FDD massive MIMO systems: Open-loop and closed-loop
training with memory,"
{\it IEEE J. Sel. Topics Signal Process.}, vol. 8, pp. 802 -- 814, Oct. 2014.

\bibitem{Noh&Zol&Sung&Love:14JSTSP} S. Noh, M. Zoltowski, Y. Sung and D. J. Love, ``Pilot beam pattern design for channel estimation in massive MIMO systems,"
{\it IEEE J. Sel. Topics Signal Process.}, vol. 8, pp. 787 -- 801, Oct. 2014.

\bibitem{So&Kim&Lee&Sung:14SPL} J. So, D. Kim, Y. Lee, and Y. Sung, ``Pilot signal design for massive MIMO systems: A received signal-to-noise-ratio-based approach,"
{\it accepted to IEEE Signal Process. Lett.}, 2014.

\bibitem{Bajwa&Haupt&Sayeed&Nowak:10IEEE} W. U. Bajwa, J. Haupt, A. M. Sayeed, and R. Nowak, ``Compressed channel sensing: a new approach to estimating sparse multipath channels,"
{\it Proc. IEEE.}, vol. 98, pp. 1058 -- 1076, Jun. 2010.

\bibitem{Taubock&Hlawatsch:08ESP}G. Taub\"ock and F. Hlawatsch, ``Compressed sensing based estimation of doubly selective channels using a sparsity-optimized basis expansion," in
{\it Proc. Eur. Signal Process. Conf.}, (Switzerland), Aug. 2008.

\bibitem{Wang&Lan&PyoEtAl:09JSAC} J. Wang, Z. Lan, C. Pyo, T. Baykas, C. Sum, M. Rahman, J. Gao, R. Funada, F. Kojima, H. Harada and Shuzo Kato, ``Beam codebook based beamforming protocol for multi-Gbps millimeter-wave WPAN systems,"
{\it IEEE Trans. J. Sel. Areas Commun.}, vol. 27, no. 8, pp. 1390 -- 1399, 2009.

\bibitem{Hur&Kim&LoveEtAl:13COM} S. Hur, T. Kim, D. Love, J. Krogmeier, T. Thomas and A. Ghosh, ``Millimeter wave beamforming for wireless backhaul and access in small cell networks,"
{\it IEEE Trans. Commun.}, vol. 61, no. 10, pp. 4391 -- 4403, 2013.

\bibitem{Tong&Sadler&Dong:04SPM} L. Tong, B. Sadler, and M. Dong, ``Pilot-assisted wireless transmissions: General model, design criteria, and signal processing,"
{\it IEEE Signal Processing Mag.}, vol. 21, no. 6, pp. 12 -- 25, 2004.

\bibitem{PutermanBook} M. L. Puterman, {\it Markov decision processes: Discrete stochastic dynamic programming.}
 John Wiley $\&$ Sons, 1994.

\bibitem{Seo&Sung&Lee&Kim:15ICCsub} J. Seo, Y. Sung, G. Lee, and D. Kim, ``Pilot beam sequence design for channel estimation in millimeter-wave MIMO systems: A POMDP framework,"
{\it submitted to ICC 2015}, Sep. 2014. Available at http://arxiv.org/abs/1409.8434.

\bibitem{Sayeed:02SP} A. M. Sayeed,
``Deconstructing multiantenna fading channels,"
  {\it IEEE Trans. Signal Process.}, vol. 50, pp. 2563 -- 2579, Oct. 2002.

\bibitem{Sayeed&Raghavan:07JSTSP} Akbar M. Sayeed and Vasanthan Raghavan, ``Maximizing MIMO capacity in sparse multipath with reconfigurable antenna arrays,"
{\it IEEE J. Sel. Topics Signal Process.}, vol. 1, pp. 156 -- 166, Jun. 2007.

\bibitem{Poor:book} H. V. Poor, ``An Introduction to Signal Detection and Estimation,"
New York: Springer, 1994.

\bibitem{Sung&Tong&Poor:06IT} Y. Sung, L. Tong, and H. V. Poor, ``Neyman-Pearson detection of Gauss-Markov signals in noise: Closed-form error exponent and properties,"
{\it IEEE Trans. Inf. Theory}, vol. 52, pp. 1354 -- 1365, Apr. 2006.

\bibitem{Smallwood&Sondic:73OR} R. D. Smallwood and E. J. Sondik, ``The optimal control of partially observable Markov processes over a finite horizon," in
{\it Operations Research}, vol. 21, pp. 1071--1088, 1973.

\bibitem{Monahan:82MS} E. Monahan, ``State of the Art—A Survey of Partially Observable Markov Decision Processes: Theory, Models, and Algorithms,"
{\it Management Science. }, vol. 28, no. 1, pp. 1 -- 16, 1982.

\bibitem{Ross70applied} S. M. Ross, {\it Applied probability models with optimization applications.}
{\it} Courier Dover Publications, 1970.

\bibitem{Cassandra&Anthony:97UAI}A. Cassandra, M. L. Littman and N. L. Zhang, ``Incremental pruning: A simple, fast, exact method for partially observable Markov decision processes,"
{\it Proc. Thirteenth Ann. Conf. on Uncertainty in Artificial Intelligence.}, pp. 54 -- 61, Morgan Kaufmann Publishers Inc., 1997.

\bibitem{Pineau&Gordon:03IJCAI} J. Pineua, G. Gordon, and S. Thrun, ``Point-based value iteration: An anytime algorithm for POMDPs," in
{\it Proc. Int. Joint Conf. Artificial Intelligence}, vol. 3, pp. 1025 -- 1032, 2003.

\bibitem{Kurniawati&Hsu:08Robtics} H. Kurniawati and D.Hsu and W. Lee, ``SARSOP: Efficient Point-Based POMDP Planning by Approximating Optimally Reachable Belief Spaces," in
{\it Proc. Robot. Sci. Syst}, 2008.

\bibitem{Smith&Simmons:12arXiv} T. Smith and R. Simmons, ``Point-based POMDP algorithms: Improved analysis and implementation,"
{\it arXiv preprint arXiv:1207.1412 },  2012.

\end{thebibliography}


\end{document}